\documentclass[a4paper,10pt]{article}

 \usepackage{amsopn}
 \usepackage{authblk}
 \usepackage{hyperref}
 \usepackage{rotating}
 \usepackage{overpic}
 \usepackage{enumerate}
 \usepackage{varwidth}
 \usepackage{subfig}
\usepackage{booktabs}
\usepackage{algorithm}
\usepackage[noend]{algpseudocode}
\usepackage{amsthm,amsmath,amssymb}
\usepackage{environ,tabularx}
 \usepackage{color} 
\usepackage{lineno}

\newtheorem{lemma}{Lemma}

\newtheorem{theorem}{Theorem}
\newtheorem{definition}{Definition}
\newtheorem{corollary}{Corollary}

\newtheorem{claim}{Claim}
\newtheorem{observation}{Observation}[section]

\usepackage{multirow}
\usepackage[table,xcdraw]{xcolor}

\usepackage[utf8]{inputenc}

\newcommand{\Oh}{\mathcal{O}}
\newcommand{\NP}{\mbox{NP}}
\newcommand{\PP}{\mbox{P}}
\newcommand{\xig}{\xi_{\mbox{\tiny $G$}}}

\newcommand{\CX}{{\mathcal X}}

\newcommand{\mc}{{\rm mc}}

\providecommand{\keywords}[1]
{ \medskip  
  \noindent
  \small	
  \textbf{\textit{Keywords---}} #1
}

\begin{document}

\title{Computing the Largest Bond and the Maximum Connected Cut of a Graph
\thanks{This work is partially supported by 
JST CREST JPMJCR1401, and JSPS KAKENHI grant numbers JP17H01788, JP16K16010, JP17K19960, and JP19K21537, and by São Paulo Research Foundation (FAPESP) grant number 2015/11937-9, and Rio de Janeiro Research Foundation (FAPERJ) grant number E-26/203.272/2017, and by National Council for Scientific and Technological Development (CNPq-Brazil) grant numbers 308689/2017-8, 425340/2016-3, 313026/2017-3, 422829/2018-8, 303726/2017-2.
The Japanese authors thank Akitoshi Kawamura and Yukiko Yamauchi for giving an opportunity to discuss in  the Open Problem Seminar in Kyushu University, Japan. Preliminary versions of this paper appeared in \cite{eto_et_al:LIPIcs:2019:11474} and \cite{duarte_et_al:LIPIcs:2019:11473}.}
}

\author[1]{Gabriel L. Duarte}
\author[2]{Hiroshi Eto}
\author[3]{Tesshu Hanaka}
\author[4]{Yasuaki Kobayashi}
\author[4]{Yusuke Kobayashi}
\author[5]{Daniel Lokshtanov}
\author[6]{Lehilton L. C. Pedrosa}
\author[6]{Rafael C. S. Schouery}
\author[1]{U\'everton~S.~Souza}

\affil[1]{Fluminense Federal University, Niter\'oi, RJ, Brazil
\thanks{\texttt{\{gabrield, ueverton\}@ic.uff.br}}}
\affil[2]{Kyushu University, Fukuoka, Japan 
\thanks{\texttt{h-eto@econ.kyushu-u.ac.jp}}}
\affil[3]{Chuo University, Tokyo, Japan
\thanks{\texttt{hanaka.91t@g.chuo-u.ac.jp}}}
\affil[4]{Kyoto University, Kyoto, Japan 
\thanks{\texttt{kobayashi@iip.ist.i.kyoto-u.ac.jp, yusuke@kurims.kyoto-u.ac.jp }}}
\affil[5]{University of California Santa Barbara, CA, USA
\thanks{\texttt{daniello@ucsb.edu}}}
\affil[6]{University of Campinas, SP, Brazil
\thanks{\texttt{\{lehilton, rafael\}@ic.unicamp.br}}}

\date{}


\maketitle

\begin{abstract}
The cut-set $\partial(S)$ of a graph $G=(V,E)$ is the set of edges that have one endpoint in $S\subset V$ and the other endpoint in $V\setminus S$, and whenever $G[S]$ is connected, the cut $[S,V\setminus S]$ of $G$ is called a connected cut. A bond of a graph $G$ is an inclusion-wise minimal disconnecting set of $G$, i.e., bonds are cut-sets that determine cuts $[S,V\setminus S]$ of $G$ such that $G[S]$ and $G[V\setminus S]$ are both connected. 
Contrasting with a large number of studies related to maximum cuts, there exist very few results regarding the largest bond of general graphs.
In this paper, we aim to reduce this gap on the complexity of computing the largest bond, and the maximum connected cut of a graph.
Although cuts and bonds are similar, we remark that computing the largest bond and the maximum connected cut of a graph tends to be harder than computing its maximum cut.
We show that it does not exist a constant-factor approximation algorithm to compute the largest bond, unless $\PP = \NP$. Also, we show that {\sc Largest Bond} and {\sc Maximum Connected Cut} are \NP-hard even for planar bipartite graphs, whereas \textsc{Maximum Cut} is trivial on bipartite graphs and polynomial-time solvable on planar graphs. In addition, we show that {\sc Largest Bond} and {\sc Maximum Connected Cut} are \NP-hard on split graphs, and restricted to graphs of clique-width $w$ they can not be solved in time $f(w)\times n^{o(w)}$ unless the Exponential Time Hypothesis fails, but they can be solved in time $f(w)\times n^{O(w)}$.
Finally, we show that both problems are fixed-parameter tractable when parameterized by the size of the solution, the treewidth, and the twin-cover number.

\end{abstract}

\section{Introduction}

Let $G=(V,E)$ be a simple, connected, undirected graph.  A \emph{disconnecting set} of $G$ is a set of edges $F\subseteq E(G)$ whose removal disconnects $G$. The edge-connectivity of $G$ is $\kappa'(G) = \min\{|F|: \mbox{$F$ is a disconnecting set of $G$}\}$. A cut $[S,T]$ of $G$ is a partition of $V$ into two subsets $S$ and $T=V\setminus S$. The cut-set $\partial(S)$ of a cut $[S,T]$ is the set of edges that have one endpoint in $S$ and the other endpoint in $T$; these edges are said to cross the cut. In a connected graph, each cut-set determines a unique cut.
Note that every cut-set is a disconnecting set, but the converse is not true.
An inclusion-wise minimal disconnecting set of a graph is called a \emph{bond} (or a \emph{minimal cut}). It is easy to see that every bond is a cut-set, but there are cut-sets that are not bonds. More precisely, a nonempty set of edges $F$ of $G$ is a bond if and only if $F$ determines a cut $[S,T]$ of $G$ such that $G[S]$ and $G[T]$ are both connected. Let~$s,t\in V(G)$. An $st$-bond of $G$ is a bond whose removal disconnects $s$ and $t$. 

A minimum (maximum) cut of a graph $G$ is a cut with cut-set of minimum (maximum) size. Every minimum cut is a bond, thus a minimum bond is also a minimum cut of $G$, and it can be found in polynomial time using the classical Edmonds–Karp algorithm~\cite{EdmondsKarpAlgorithm}. Besides that, minimum $st$-bonds are well-known structures, since they are precisely the $st$-cuts involved in the Gomory-Hu trees~\cite{GomoryHu}.

\textsc{Maximum Cut} is one of the most fundamental problems in theoretical computer science. Given a graph $G$ and an integer $k$, the problem asks for a subset $S$ of vertices such that $|\partial(S)|\geq k$. Whenever the subgraph of $G$ induced by $S$ is connected, the cut $[S, V\setminus S]$ is called {\em connected}. 
Recall that every bond is a connected cut, but the converse is not true.

In this paper, we are interested in the complexity aspects of the following problem.

\medskip
\noindent    \fbox{
        \parbox{.95\textwidth}{
\noindent
\textsc{Maximum Connected Cut}\\
\noindent
\textbf{Instance}: A graph $G=(V,E)$; a positive integer $k$.\\
\noindent
\textbf{Question}: Is there a proper subset $S\subset V(G)$ such that  $G[S]$ is connected and $|\partial(S)|\geq k$?
}
}

\medskip
\noindent    \fbox{
        \parbox{.95\textwidth}{
\noindent
\textsc{Largest Bond}\\
\noindent
\textbf{Instance}: A graph $G=(V,E)$; a positive integer $k$.\\
\noindent
\textbf{Question}: Is there a proper subset $S\subset V(G)$ such that  $G[S]$ and $G[V\setminus S]$ are both connected and $|\partial(S)|\geq k$?
}
}

\medskip

We also consider the versions of both problems where the removal of the cut-set must disconnect a given pair of vertices.

\noindent    \fbox{
        \parbox{.95\textwidth}{
\noindent
\textsc{Maximum Connected $st$-Cut}\\
\noindent
\textbf{Instance}: A graph $G=(V,E)$; vertices $s,t\in V(G)$; a positive integer $k$.

\noindent
\textbf{Question}: Is there a proper subset $S\subset V(G)$ with $s\in S$ and $t\notin S$, such that  $G[S]$ is connected and $|\partial(S)|\geq k$?
}
}

\medskip

\noindent    \fbox{
        \parbox{.95\textwidth}{
\noindent
{\sc \textsc{Largest $st$-Bond}}

\noindent
\textbf{Instance}: A graph $G=(V,E)$; vertices $s,t\in V(G)$; a positive integer $k$.


\noindent
\textbf{Question}: Is there a proper subset $S\subset V(G)$ with $s\in S$ and $t\notin S$, such that $G[S]$ and $G[V\setminus S]$ are both connected and $|\partial(S)|\geq k$?
}
}

\medskip

Such problems can be seen as variants of \textsc{Maximum Cut}, which was shown to be NP-hard in Karp's seminal work~\cite{Karp1972}. To overcome this intractability, a lot of researches have been done from various view points, such as  approximation algorithms~\cite{goemans1995semidef}, fixed-parameter tractability~\cite{Raman2007}, and special graph classes~\cite{bodlaender2000,Diaz2007,Guruswami1999,Hadlock,Orlova1972}.

The \textsc{Maximum Connected Cut} problem was defined in \cite{Haglin} and it is known to be NP-complete even on planar graphs \cite{HajiaghayiKMPS15}. Regarding bonds on planar graphs, a folklore theorem states that if $G$ is a connected planar graph, then a set of edges is a cycle in $G$ if and only if it corresponds to a bond in the dual graph of~$G$~\cite{LivroPlanarBond}. Note that each cycle separates the faces of $G$ into the faces in the interior of the cycle and the faces of the exterior of the cycle, and the duals of the cycle edges are exactly the edges that cross from the interior to the exterior~\cite{oxley2006matroid}. Consequently, the girth of a planar graph equals the edge connectivity of its dual~\cite{CHO20072456}.
Although cuts and bonds are similar, computing the largest bond of a graph seems to be harder than computing its maximum cut. {\sc Maximum Cut} is \NP-hard in general~\cite{garey2002computers}, but becomes polynomial for planar graphs~\cite{Hadlock}. On the other hand, finding a longest cycle in a planar graph is \NP-hard, implying that finding a largest bond of a planar multigraph (or of a simple edge-weighted planar graph) is \NP-hard. In addition, it is well-known that if a simple planar graph is $3$-vertex-connected, then its dual is a simple planar graph. In~1976, Garey, Johnson, and Tarjan~\cite{PlanarHamiltonian} proved that the problem of establishing whether a $3$-vertex-connected planar graph is Hamiltonian is \NP-complete, thus, as also noted by Haglin and Venkatesan~\cite{Haglin},  finding the largest bond of a simple planar graph is also \NP-hard, contrasting with the polynomial-time solvability of {\sc Maximum Cut} on planar graphs. Recently, Chaourar proved that \textsc{Largest Bond} can be solved in polynomial time on series parallel graphs and graphs without $K_5\setminus e$ as a minor in \cite{Chaourar2017,Chaourar2019}.

Computing the largest bond of a graph $G$ corresponds to compute the maximum minimal cut of $G$. Graph problems about finding a maximum minimal (or minimum maximal) solutions such as \textsc{Maximum Minimal Vertex Cover} \cite{Boria2015,Zehavi2017}, \textsc{Maximum Minimal Dominating Set} \cite{Bazgan2018}, \textsc{Maximum Minimal Edge Cover} \cite{Khoshkhah2019}, \textsc{Maximum Minimal Separator} \cite{Hanaka2017}, \textsc{Minimum Maximal matching} \cite{garey2002computers,Yannakakis1980}, and \textsc{Minimum Maximal Independent Set} \cite{Demange1999}, have been long studied.

From the point of view of parameterized complexity, it is well known that {\sc Maximum Cut} can be solved in FPT time when parametrized by the size of the solution~\cite{mahajan1999parameterizing}, and since every graph has a cut with at least half the edges~\cite{erdos1965some}, it follows that it has a linear kernel. Concerning approximation algorithms, a $1/2$-approximation algorithm can be obtained by randomly partitioning the set vertices into two parts, which induces a cut-set whose expected size is at least half of the number of edges~\cite{mitzenmacher_upfal_2005}. The best-known result is the seminal work of Goemans and Williamson~\cite{goemans1995semidef}, who gave a $0.878$-approximation based on semidefinite programming. This has the best approximation factor unless the Unique Games Conjecture fails~\cite{khot2007hard}.
To the best of our knowledge, there is no algorithmic study regarding the parameterized complexity of computing the largest bond of a graph as well as the approximability of the problem.
%
%
Observe that a bond induces a feasible solution of {\sc Maximum Connected Cut}, but not the other way around, since $G[T]$ may be disconnected. Indeed, the size of a largest bond can be arbitrarily smaller than the size of the maximum connected cut; take, e.g., a star with $n$ leaves.
For {\sc Maximum Connected Cut} on general graphs, there exists a $\Omega(1/\log n)$-approximation~\cite{GandhiHKPS18}, where $n$ is the number of vertices. Also, there is a constant-factor approximation with factor~$1/2$ for graphs of bounded treewidth~\cite{ShenLN18}, and a polynomial-time approximation scheme for graphs of bounded genus~\cite{HajiaghayiKMPS15}.

Recently, Saurabh and Zehavi~\cite{SaurabhZ19} considered a generalization of {\sc Maximum Connected Cut}, named {\sc Multi-Node Hub}. In this problem, given numbers $l$ and $k$, the objective is to find a cut $[S, T]$ of $G$ such that $G[S]$ is connected, $|S| = l$ and $|\partial(S)| \ge k$.
They observed that the problem is $W[1]$-hard when parameterized on~$l$, and gave the first parameterized algorithm for the problem with respect to the parameter~$k$. We remark that the $W[1]$-hardness also holds for {\sc Largest Bond} parameterized by $|S|$.

Since every nonempty bond determines a cut $[S,T]$ such that $G[S]$ and $G[T]$ are both connected, every bond of $G$ has size at most $|E(G)|-|V(G)|+2$. A graph~$G$ has a bond of size $|E(G)|-|V(G)|+2$ if and only if $V(G)$ can be partitioned into two parts such that each part induces a tree. Such graphs are known as \emph{Yutsis graphs}. The set of planar Yutsis graphs is exactly the dual class of  Hamiltonian planar graphs.
According to Aldred, Van Dyck, Brinkmann, Fack, and McKay~\cite{aldred2009graph}, cubic Yutsis graphs appear in the quantum theory of angular momenta as a graphical representation of general recoupling coefficients. They can be manipulated following certain rules in order to generate the so-called summation formulae for the general recoupling coefficient (see \cite{biedenharn1981racah,VANDYCK20071506,yutsis1962mathematical}).

There are very few results about the largest bond size in general graphs. In 2008, Aldred, Van Dyck, Brinkmann, Fack, and McKay~\cite{aldred2009graph} showed that if a Yutsis graph is regular with degree $3$, the partition of the vertex set from the largest bond will result in two sets of equal size. In 2015, Ding, Dziobiak and Wu~\cite{ding2016large} proved that any simple $3$-connected graph $G$ will have a largest bond with size at least $\frac{2}{17}\sqrt{\log n}$, where $n=|V(G)|$. In 2017, Flynn~\cite{flynn2017largest} verified the conjecture that any simple $3$-connected graph $G$ has a largest bond with size at least $\Omega(n^{\log_32})$ for a variety of graph classes including planar graphs.

Even though there are many important applications of  \textsc{Maximum Connected Cut} and \textsc{Largest Bond} such as image segmentation \cite{Vicente2008}, forest planning \cite{Carvajal2013}, and  computing a market splitting for electricity markets \cite{Grimm2019}, the known results are much fewer than those for \textsc{Maximum Cut} due to the difficult nature of simultaneously maximizing its size and handling the connectivity of a cut.

\subsection{Our contributions}

In this paper, we complement the state of the art on the problems of computing the largest bond and the maximum connected cut of a graph.
Preliminarily, we present general reductions that allows us to observe that {\sc Largest Bond} and {\sc Maximum Connected Cut} are \NP-hard for several graph classes for which {\sc Maximum Cut} is \NP-hard. Using this framework, we are able to show that {\sc Largest Bond} and {\sc Maximum Connected Cut} on graphs of clique-width $w$ cannot be solved in time $f(w)\times n^{o(w)}$ unless the ETH fails.
We also prove that both \textsc{Maximum Connected Cut} and \textsc{Largest Bond} are NP-complete even on planar bipartite graphs.
Interestingly, although \textsc{Maximum Cut} can be solved in polynomial time on planar graphs \cite{Hadlock,Orlova1972} and it is trivial on bipartite graphs, our both problems are intractable even on the intersection of these classes. Also, we show that these problems are NP-complete on split graphs.
Moreover, we show that {\sc Largest Bond} does not admit a constant-factor approximation algorithm, unless $\PP = \NP$, and thus is asymptotically harder to approximate than {\sc Maximum Cut}.

To tackle this difficulty, we study both problems from the perspective of the parameterized complexity.
%
Using win/win approaches, we consider the strategy of preprocessing the input in order to bound the treewidth of the resulting instance. 
After that, we give $O^*({2}^{O(tw\log tw)})$-time algorithms for both problems\footnote{The $O^*(\cdot)$ notation suppresses polynomial factors in the input size.}, where $tw$ is the tree-width of the input graph.
Moreover, we can improve this running time using the rank-based approach \cite{Bodlaender2015} to $O^*(c^{tw})$ for some constant $c$ and using the Cut \& Count technique \cite{Cygan2011} to $O^*(3^{{tw}})$ for \textsc{Maximum Connected Cut} and $O^*(4^{tw})$ for \textsc{Largest Bond}, using randomization.
Let us note that our result generalizes the polynomial time algorithms for \textsc{Largest Bond} on series parallel graphs and graphs without $K_5\setminus e$ as a minor due to Chaourar~\cite{Chaourar2017,Chaourar2019} since such graphs are tree-width bounded~\cite{Robertson1986}.
Based on these algorithms, we give $O^*(2^{O(k)})$-time algorithms for both problems.
Also, we remark that  the problems do not admit polynomial kernels, unless NP $\subseteq$ coNP/poly.

Finally, we consider different structural graph parameters.
We design tight (assuming ETH) XP-time algorithms for both problems when parameterized by clique-width $cw$.
Also, we give $O^*(2^{2^{tc}+tc})$-time and $O^*(2^{tc}3^{2^{tc}})$-time FPT algorithms for \textsc{Maximum Connected Cut} and \textsc{Largest Bond}, respectively, where $tc$ is the minimum size of a twin-cover of the input graph.

\section{Intractability results}

In this section, we discuss aspects of the hardness of computing the largest bond and the maximum connected cut of a graph.
Notice that {\sc Largest Bond} and {\sc Maximum Connected Cut} are Turing reducible to {\sc Largest $st$-Bond} and  {\sc Maximum Connected $st$-Cut}, respectively. Therefore, the hardness of {\sc Largest Bond} and {\sc Maximum Connected Cut}  presented in this section also holds for {\sc Largest $st$-Bond} and  {\sc Maximum Connected $st$-Cut} as well, unless $P=NP$.

\smallskip

Next, we present a general framework for reducibility from {\sc Maximum Cut} to {\sc Largest Bond}, by defining a special graph operator $\psi$ such that {\sc Maximum Cut} on a graph class~$\mathcal{F}$ is reducible to {\sc Largest Bond} on the image of $\mathcal{F}$ via~$\psi$. An interesting particular case occurs when $\mathcal{F}$ is closed under $\psi$ (for instance, chordal graphs are closed under $\psi$).

\begin{definition}
Let $G$ be a graph and let $n=V(G)$. The graph $\psi(G)$ is constructed as follows: 
\begin{enumerate}
\item[(i)] create $n$ disjoint copies $G_1, G_2, \ldots, G_n$ of $G$;
\item[(ii)] add vertices $v_a$ and $v_b$;
\item[(iii)] add an edge between $v_a$ and $v_b$;
\item[(iv)] add all  possible edges between $V(G_1 \cup G_2 \cup \ldots \cup G_n)$ and~$\{v_a,v_b\}$.
\end{enumerate}
\end{definition}

\begin{definition}
A set of graphs~$\mathcal{G}$ is closed under operator $\psi$ if whenever $G \in \mathcal{G}$, then $\psi(G)\in \mathcal{G}$.
\end{definition}


\begin{theorem}\label{framework}
{\sc Largest Bond} is \NP-complete for any graph class $\mathcal{G}$ such that:
\begin{itemize}
\item $\mathcal{G}$ is closed under operator $\psi$; and 
\item {\sc Maximum Cut} is \NP-complete for graphs in $\mathcal{G}$.
\end{itemize}
\end{theorem}

\begin{proof}
Let $G\in \mathcal{G}$, $n = |V(G)|$, and $H=\psi(G)$. By~(i), $H\in\mathcal{G}$. Suppose that $G$ has a cut $[S,V(G)\setminus S]$ of size $k$, and let $S_1$, $S_2$, \ldots, $S_n$ be the copies of $S$ in $G_1, G_2, \ldots, G_n$, respectively. If $S'=\{v_a\}\cup S_1\cup S_2 \cup \ldots \cup S_n$, then $[S',V(H)\setminus S']$ defines a bond $\partial(S')$ of~$H$ of size at least $nk+n^2+1$.
Conversely, suppose $H$ has a bond $\partial(S')$ of size at least $nk+n^2+1$. We consider the following cases: $(a)$~If $\{v_a, v_b\}\subseteq S'$, then for all copies $G_i$ but one we have $V(G_i)\subseteq S'$, as otherwise the graph induced by $V(H)\setminus S'$ would not be connected, and~$\partial(S')$ would not be a bond. Thus, $V(H)\setminus S' \subseteq V(G_j)$ for some $j$, then the size of $\partial(S')$ is smaller than $nk+n^2+1$, a contradiction. $(b)$~If $v_a\in S'$ and $v_b\notin S'$,
then  $\{v_a,v_b\}$ is incident with exactly $n^2+1$ edges crossing $[S',V(H)\setminus S']$, which implies that at least one copy $G_i$ has $k$ or more edges crossing $[S',V(H)\setminus S']$. Therefore, $G$ has a cut of size at least $k$. 
\end{proof}

\begin{corollary}\label{subclasses_NPc}
{\sc Largest Bond} is \NP-complete for the following classes:
\begin{enumerate}
  \item chordal graphs;
  \item co-comparability graphs;
  \item $P_5$-free graphs;
  \item AT-free graphs.
\end{enumerate}
\end{corollary}

\begin{proof}
Bodlaender and Jansen~\cite{bodlaender2000complexity} proved that {\sc Maximum Cut} is \NP-complete when restricted to split and co-bipartite graphs. Since split graphs are chordal and co-bipartite graphs are $P_5$-free, AT-free and co-comparability graphs, the \NP-completeness also holds for these classes.
Now we have to show that the classes are closed under $\psi$.

\emph{(1.)} A graph is chordal if every cycle of length at least $4$ has a chord. Let $G$ be a chordal graph. Notice that the disjoint union of $G_1, G_2, \ldots, G_n$ is also chordal. In addition, no chordless cycle of length at least $4$ may contain either $v_a$ or $v_b$ because both vertices are universal. Therefore,~$\psi(G)$ is chordal.

\emph{(2.)} A graph is a co-comparability graph if it is the intersection graph of curves from a line to a parallel line. Let $G$ be a co-comparability graph. Notice that the class of co-comparability graphs is closed under disjoint union. Thus, in order to conclude that $\psi(G)$ is co-comparability, it is enough to observe that from a representation of curves (from a line to a parallel line) of the disjoint union of $G_1, G_2, \ldots, G_n$, one can construct a representation of $\psi(G)$ by adding two concurrent lines (representing $v_a$ and $v_b$) crossing all curves.

\emph{(3.)} The disjoint union of $P_5$-free graphs is also $P_5$-free. In addition, no induced $P_5$ contains either $v_a$ or $v_b$ because both vertices are universal. Then, the class of $P_5$-free graphs is closed under~$\psi$.

\emph{(4.)} Three vertices of a graph form an asteroidal triple if every two of them are connected by a path avoiding the neighbourhood of the third.
A graph is AT-free if it does not contain any asteroidal triple. 
Since an asteroidal triple does not contain universal vertices and it is a connected subgraph, the class of AT-free graphs is closed under~$\psi$.  
\end{proof}


Now we consider a similar result for {\sc Maximum Connected Cut}.

\begin{definition}
Let $G$ be a graph and let $n=V(G)$. The graph $\phi(G)$ is constructed as follows: 
\begin{enumerate}
\item[(i)] create $n$ disjoint copies $G_1, G_2, \ldots, G_n$ of $G$;
\item[(ii)] add a new vertex $v_a$;
\item[(iv)] add exactly one edge from $v_a$ to a vertex of $G_i$, $1\leq i\leq n$.
\end{enumerate}
\end{definition}

At this point, it is easy to see that a graph $G$ has a cut $[S,V(G)\setminus S]$ of size $k$ if and only if $\phi(G)$ has a bond $\partial(S')$ of size at least $nk$. Thus, the following theorem also holds.

\begin{theorem}\label{framework2}
{\sc Maximum Connected Cut} is \NP-complete for any graph class $\mathcal{G}$ such that:
\begin{itemize}
\item $\mathcal{G}$ is closed under operator $\phi$; and 
\item {\sc Maximum Cut} is \NP-complete for graphs in $\mathcal{G}$.
\end{itemize}
\end{theorem}

\subsection{Algorithmic lower bound for clique-width parameterization}

The {\em clique-width} of a graph~$G$, denoted by~$cw(G)$, is defined as the minimum number of labels needed to construct~$G$, using the following four operations:

\begin{enumerate}
  \item Create a single vertex~$v$ with an integer label~$\ell$ (denoted by~$\ell(v)$);

  \item Take the disjoint union (i.e., co-join) of two graphs (denoted by~$\oplus$);

  \item Join by an edge every vertex labeled~$i$ to every vertex labeled~$j$ for~$i \neq j$ (denoted by~$\eta(i,j)$);

  \item Relabel all vertices with label~$i$ by label~$j$ (denoted by~$\rho(i,j)$).
\end{enumerate}

An algebraic term that represents such a construction of $G$ and uses at most~$w$ labels is said to be a {\em $w$-expression} of $G$, and the clique-width of $G$ is the minimum~$w$ for which $G$ has a $w$-expression.

Graph classes with bounded clique-width include cographs, distance-hereditary graphs, graphs of bounded treewidth, graphs of bounded branchwidth, and graphs of bounded rank-width.

In the '90s,  Courcelle, Makowsky, and Rotics~\cite{courcelle2000linear} proved that all problems expressible in MS$_1$-logic are fixed-parameter tractable when parameterized by the cli\-que-width of a graph and the logical expression size.
The applicability of this meta-theorem has made clique-width become one of the most studied parameters in parameterized complexity. However, although several problems are MS$_1$-expressible, this is not the case with {\sc Maximum Cut}.

In 2014, Fomin, Golovach, Lokshtanov and Saurabh~\cite{fomin2014almost} showed that {\sc Maximum Cut} on a graph of clique-width $w$ cannot be solved in time $f(w)\times n^{o(w)}$ for any function $f$ of $w$ unless Exponential Time Hypothesis (ETH) fails. Using operators $\psi$ and $\phi$, we are able to extend this result to {\sc Largest Bond} and {\sc Maximum Connected Cut}.

\begin{lemma}\label{lb_cw}
{\sc Largest Bond} and {\sc Maximum Connected Cut} on graphs of clique-width $w$ cannot be solved in time $f(w)\times n^{o(w)}$ unless the ETH fails.
\end{lemma}

\begin{proof}
{\sc Maximum Cut} cannot be solved in time $f(w)\times n^{o(w)}$ on graphs of clique-width $w$, unless Exponential Time Hypothesis (ETH) fails~\cite{fomin2014almost}. Therefore, by the polynomial-time reduction presented in Theorem~\ref{framework} and Theorem~\ref{framework2}, it is enough to show that the clique-width of $\psi(G)$ and $\phi(G)$ is upper bounded by a linear function of the clique-width of $G$.

If $G$ has clique-width $w\geq 2$, then the disjoint union $H_1 = G_1\oplus G_2 \oplus \ldots \oplus G_n$ has clique-width~$w$.

For $\psi(G)$, suppose that all vertices in $H_1$ have label $1$.
Now, let $H_2$ be the graph isomorphic to a $K_2$ such that $V(H)=\{v_a,v_b\}$, and $v_a,v_b$ are labeled with $2$.
In order to construct~$\psi(G)$ from $H_1\oplus H_2$ it is enough to apply the join $\eta(1,2)$. Thus, $\psi(G)$ has clique-width equals~$w$.

For $\phi(G)$, note that from a $w$-expression of $H_1$ we can obtain a $w+1$-expression of $H_1$ resulting into a labelled graph such that the neighborhood of $v_a$ are exactly the vertices with label $w+1$ and all the other vertices of $H_1$ have label $1$. Thus we can add $v_a$ with label $2$ and apply the join $\eta(1,2)$, which implies that $\phi(G)$ has clique-width at most~$w+1$. 
\end{proof}

\subsection{On planar bipartite graphs}

Although {\sc Maximum Cut} is trivial for bipartite graphs, we first observe that the same does not apply to compute the largest bond.

\begin{theorem}\label{thm:bipartite:maxmin}
\textsc{Largest Bond} is NP-complete even on planar bipartite subcubic graphs.
\end{theorem}
\begin{proof}
In~\cite{Haglin}, Haglin and Venkatesan proved that \textsc{Largest Bond} remains NP-com\-ple\-te on planar cubic graphs. Since subdivision of edges does not increase the size of the largest bond, by subdividing each edge of a planar cubic graph $G$ we obtain planar bipartite subcubic graph $G'$ such that $G$ has a bond of size $k$ if and only if $G'$ has a bond of size $k$. 
\end{proof}

\begin{theorem}\label{w1hard}
Let $G$ be a simple bipartite graph and $\ell\in \mathbb{N}$. To determine the largest bond $\partial(S)$ of $G$ with $|S|=\ell$ is $W[1]$-hard with respect to $\ell$.
\end{theorem}

\begin{proof}
From an instance $H$ of {\sc $k$-Independent Set} on regular graphs we first construct a multigraph $G'$ by adding an edge between any pair of vertices.
Finally, we obtain a simple graph $G$ by subdividing every edge of $G'$.
Notice that $H$ has an independent set of size $k$ if and only if $G$ has a bond $\partial(S)$ of size $dk+k(n-k)$ with $|S|=k + {\genfrac(){0pt}{2}{k}{2}}$, where $d$ is the vertex degree of $H$. 
\end{proof}

Now, we consider {\sc Maximum Connected Cut} on planar bipartite graphs.

\begin{figure}[tbp]
  \begin{center}
       \includegraphics[width=80mm]{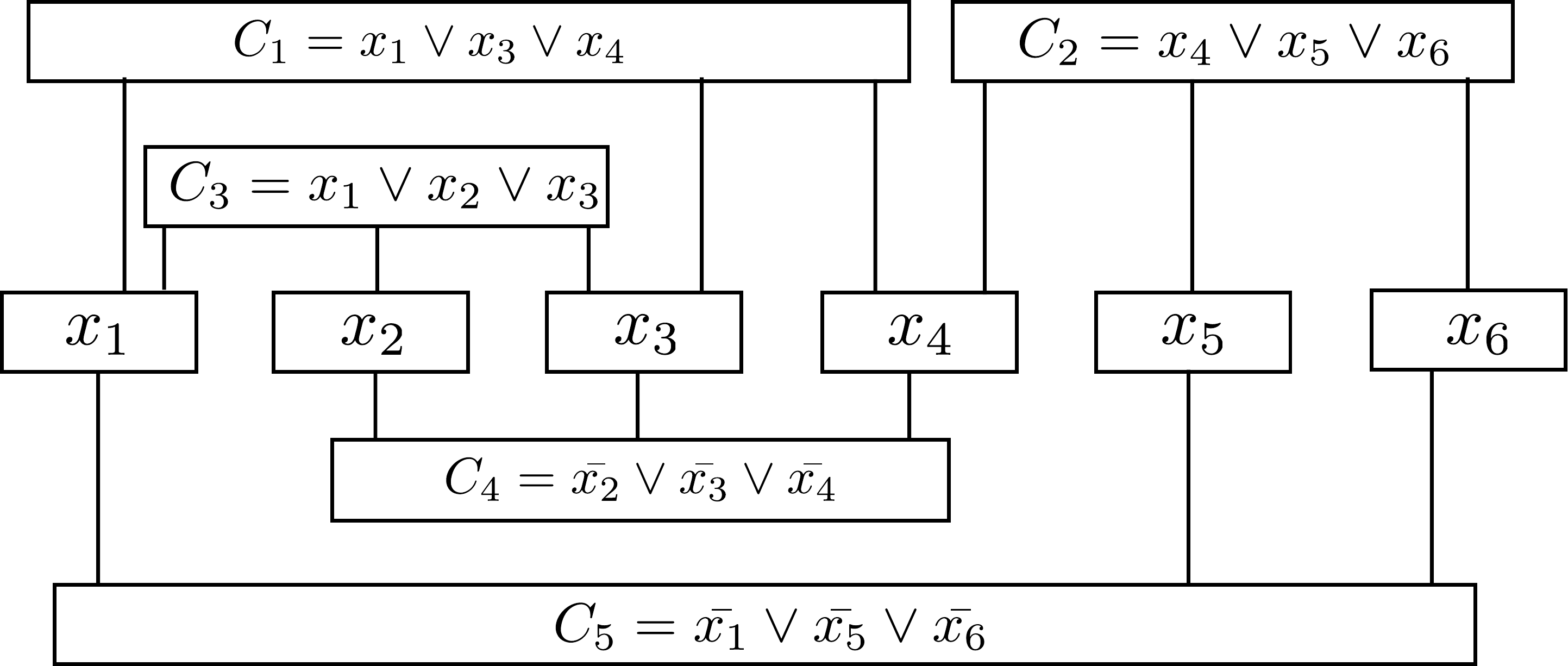}
  \end{center}
  \begin{center}
  a) {A monotone rectilinear representation of an instance of {\sc Planar Monotone Rectilinear 3-SAT}.}
  \end{center}
  \begin{center}
   \includegraphics[width=\textwidth]{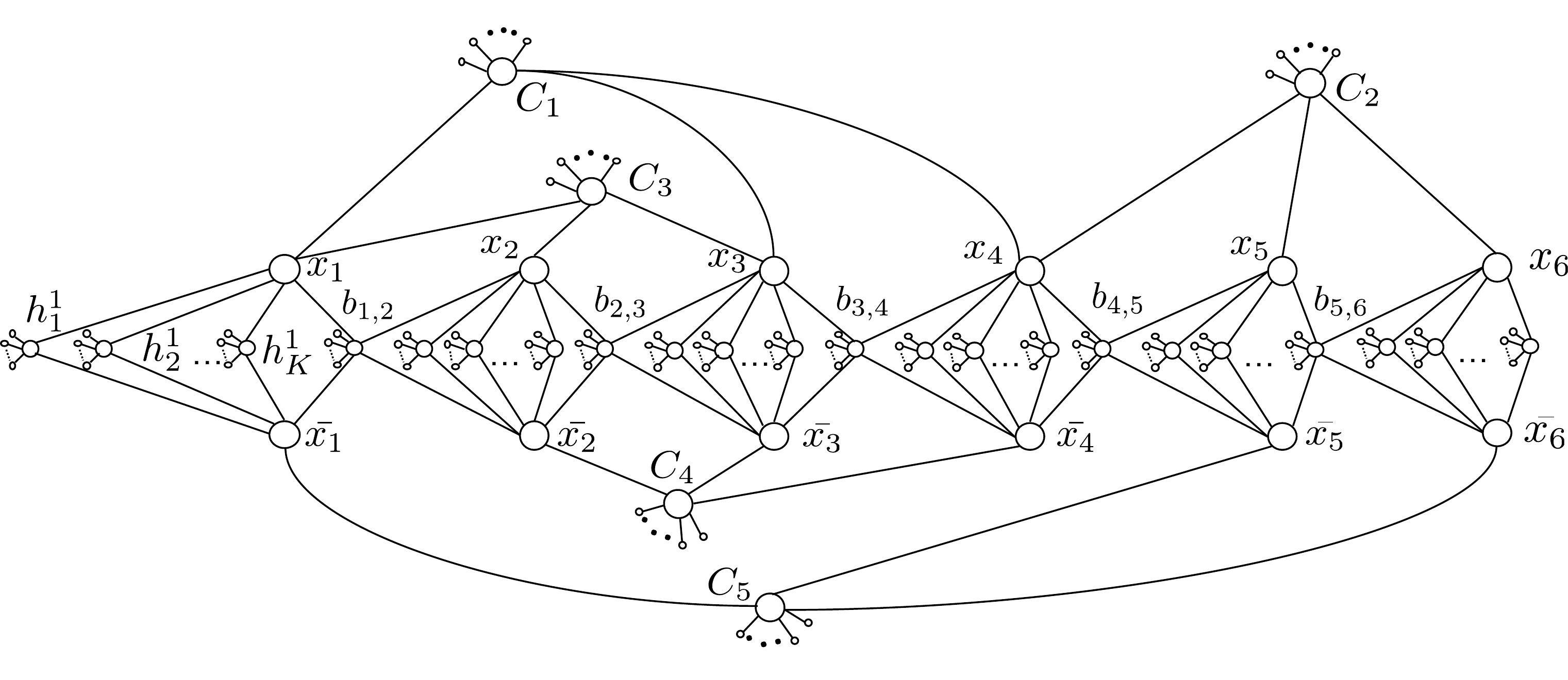}
  \end{center}
  \begin{center}
  b) {The reduced graph $H_\phi$. }
  \end{center}
  \caption{An example illustrating (a) the rectilinear representation of an formula $\phi=(x_1\lor x_3 \lor x_4)\land (x_4\lor x_5 \lor x_6)\land (x_1\lor x_2 \lor x_3)\land (\bar{x_2}\lor \bar{x_3} \lor \bar{x_4})\land(\bar{x_1}\lor \bar{x_5} \lor \bar{x_6})$ of \textsc{Planar Monotone Rectilinear 3-SAT} and (b) the reduced graph $H_\phi$.}
 \label{PM3SATtoCMConPB}
\end{figure}

\begin{theorem}\label{thm:bipartite:Connected}
\textsc{Maximum Connected Cut} is NP-complete on planar bipartite graphs.
\end{theorem}
\begin{proof}
The reduction is based on the proof of Theorem $4$ in \cite{HajiaghayiKMPS15}, which proves that \textsc{Maximum Connected Cut} is NP-hard on planar graphs.
We transform the planar reduced graph in \cite{HajiaghayiKMPS15} into our planar bipartite graph by using additional vertices, called {\em bridge vertices}.

In this proof, we reduce an instance of \textsc{Planar Monotone Rectilinear 3-SAT}, which is known to be NP-complete~\cite{Berg2009}, to a planar bipartite instance of \textsc{Maximum Connected Cut}.
An instance of \textsc{Planar Monotone Rectilinear 3-SAT} consists of a 3-CNF formula $\phi$ satisfies the following properties: (a) each clause contains either only positive literals or only negative literals, (b) the bipartite incidence graph $G_{\phi}$ is planar, and (c) $G_{\phi}$ has a {\em monotone rectilinear representation}.
In a monotone rectilinear representation of $G_{\phi}$, the variable vertices are drawn on a straight line in the order of their indices and each positive (resp., negative) clause vertex is drawn in the ``positive side'' (resp., ``negative side'') of the plane defined by the straight line (See Figure~\ref{PM3SATtoCMConPB}).

\noindent\textbf{The Reduction.}
Given a formula $\phi$  of \textsc{Planar Monotone Rectilinear 3-SAT} with a monotone rectilinear representation as in Figure~\ref{PM3SATtoCMConPB} (a), let $X=\{x_1, x_2,\dots, x_n\}$ be a set of variables and ${\mathcal C}=\{C_1, C_2, \dots, C_m\}$ be a set of clauses.
Let $K>4m^2$ and $m>2$. 
Then we create the graph $H_{\phi}=(V,E)$ as follows (see Fig.\ref{PM3SATtoCMConPB}). 
For each variable $x_i\in X$, we create two {\em literal vertices} $v(x_i)$ and $v(\bar{x}_i)$ corresponding to the literals $x_i$ and $\bar{x}_i$, respectively. 
Moreover, we add $K$ {\em helper vertices} $h^i_1, \ldots, h^i_{K}$ and connect $h^i_{k}$ to $v(x_i)$ and $v(\bar{x}_i)$ for each $k=1, \ldots, K$.
For every clause $C_j\in {\mathcal C}$, we create a {\em clause vertex} $v(C_j)$ and connect $v(x_i)$ (resp., $v(\bar{x}_i)$) to $v(C_j)$ if $C_j$ contains $x_i$ (resp., $\bar{x}_i$).
Moreover, we attach $\sqrt{K}$ pendant vertices to each $v(C_j)$.

Then we attach $K$ pendant vertices to each helper vertex  $h^i_k$. 
Finally, we add a {\em bridge vertex} $b_{i,i+1}$ with  $K$ pendant vertices that we make adjacent to each  $v(x_{i})$, $v(\bar{x}_{i})$, $v(x_{i+1})$, and $v(\bar{x}_{i+1})$ for $1\le i\le n-1$. 
We denote by $H_\phi$ the graph we obtained.
Notice that we can draw $H_\phi$ in the plane according to a monotone rectilinear representation.
Moreover, $H_\phi$ is bipartite since we only add helper and bridge vertices, which have a neighbor only in literal vertices, and pendant vertices to the planar drawing of $G_\phi$.

Clearly, this reduction can be done in polynomial time.
To complete the proof, we prove the following claim.
\begin{claim}
A formula $\phi$ is satisfiable if and only if there is a maximum connected cut of size at least $m\sqrt{K} + nK^2 + (2n-1)K +2(n-1)$ in $H_\phi$. 
\end{claim}
\begin{proof}
Let $V_X$, $V_C$, $V_H$, $V_B$, and $V_P$ be the set of literal vertices, clause vertices, helper vertices, bridge vertices, and pendant vertices, respectively.

\noindent ($\Rightarrow$) We are given a satisfiable assignment $\alpha$ for $\phi$. 
For $\alpha$, we denote a true literal by $l_i$. We also call $v(l_i)$ a {\em true literal vertex}.
Let $S=\bigcup_{i=1}^n \{v(l_i)\}\cup V_C\cup V_H\cup V_B$.
That is, $S$ consists of  the set of true literal vertices, all the clauses vertices, all helper vertices and bridge vertices.
Observe that the induced subgraph by $S$ is connected.
This follows from the facts that each clause has at least one true literal and literal vertices are connected by bridge vertices.

Finally, we show that $|\partial(S)|\ge m\sqrt{K} + nK^2 + (2n-1)K +2(n-1)$.
Since each clause vertex has $\sqrt{K}$  pendant vertices and each helper vertex $K$ pendant vertices, there are $m\sqrt{K} + nK^2$ cut edges.
Moreover, each bridge vertex has $K$ cut edges incident to its pendant vertices and two cut edges incident to literal vertices not in $S$.
Finally, since either $v(x_i)$ or $v(\bar{x_i})$ is not in $S$, there are $nK$ cut edges between literal vertices and helper vertices.
Therefore, we have $|\partial(S)|\ge m\sqrt{K} + nK^2 + (n-1)(K+2) + nK = m\sqrt{K} + nK^2 + (2n-1)K + 2(n-1)$.

\medskip
\noindent ($\Leftarrow$)
We are given a connected cut $S$ in $H_\phi$ such that $|\partial(S)|\ge m\sqrt{K} + nK^2 + (2n-1)K +2(n-1)$.
Here, we assume without loss of generality that $S$ is an optimal connected cut of $H_\phi$.
Suppose, for contradiction, that neither of $v(x_i)$ and $v(\bar{x}_i)$ is contained in $S$ for some variable $x_i$. 
Then, all helper vertices $h^i_k$ cannot be contained in $S$ due to the connectivity of $S$. 
There are $m\sqrt{K} + 3m + 2nK + (K + 4)(n-1)$ edges except for those between helper vertices and its pendant vertices.
Thus, it follows that $|\partial(S)| \le m\sqrt{K} + 3m + 2nK + (K + 4)(n-1) + (n-1)K^2$. 
Since $K>4m^2$ and $m>2$, this contradicts the assumption that $|\partial(S)|\ge m\sqrt{K} + nK^2 + (2n-1)K +2(n-1)$.
Thus, at least one literal vertex must be contained in $S$ for each $x_i$.

Next, we show that every helper vertex must be contained in $S$.
Suppose a helper vertex $h^i_k$ is not contained in $S$.
Then, all $K$ pendant vertices attached to $h^i_k$ is not contained in $S$ due to the connectivity of $S$.
Since at least one literal vertex of $x_i$ is contained in $S$, we can increase the size of the cut by moving $h^i_k$ to $S$, contradicting the optimality of $S$.
Therefore, we assume that every helper vertex is contained in $S$.
Similar to helper vertices, we can prove that every bridge vertex is contained in $S$.

Then, we observe that exactly one literal vertex must be contained in $S$ for each~$x_i$.
Suppose that both $v(x_i)$ and $v(\bar{x}_i)$ are contained in  $S$ for some $x_i$.
Since all helper vertices and bridge vertices are contained in $S$, we may increase the size of the cut by moving either of $v(x_i)$ or $v(\bar{x}_i)$ to $V \setminus S$.
However, there are some issues we have to consider carefully.
Suppose that $v(x_i)$ is moved to $V \setminus S$.
Then, some clause vertices $v(C_j)$ in $S$ can be disconnected in $G[S]$.
If so, we also move $v(C_j)$ together with its pendant vertices to $V \setminus S$.
Since there are at least $K + 1$ cut edges newly introduced but at most $m(\sqrt{K} + 3)$ edges removed from the cutset, the size of the cutset is increased, also contradicting the optimality of $S$.

Finally, we show that every clause vertex is, in fact, contained in $S$.
Suppose that $|S \cap V_C| = m' < m$.
If $v(C_j)$ is not in $S$, its pendant vertices are also not in $S$.
Due to the optimality of $S$, the pendant vertices of every helper vertex and every bridge vertex is in $V \setminus S$.
Thus, we have $|\partial(S) \cap \partial(V_H \cup V_B)| = nK(K+1) + (K+2)(n-1) = nK^2 + (2n-1)K+2(n-1)$.
Therefore, $|\partial(S)| = |\partial(S) \cap \partial(V_C)| +  |\partial(S) \cap \partial(V_H \cup V_B)| \le m'\sqrt{K} +3(m-m')+ nK^2 + (2n-1)K +2(n-1) < m\sqrt{K} + nK^2 + (2n-1)K +2(n-1)$ as $K > 4m^2$.
This is also contradicting to the assumption of the size of the cut.

To summarize, exactly one literal vertex is in $S$ for each variable and every clause vertex is in $S$.
Since $G[S]$ is connected, every clause vertex is adjacent to a literal vertex included in $S$.
Given this, we can obtain a satisfying assignment for $\phi$. 
\end{proof}
This completes the proof of Theorem \ref{thm:bipartite:Connected}. 
\end{proof}

\subsection{On split graphs}\label{sec:split}

\begin{theorem}\label{thm:split:connected}
\textsc{Maximum Connected Cut} is NP-complete on split graphs.
\end{theorem}
\begin{proof} 
We reduce the following problem called {\sc Exact 3-cover}, which is known to be NP-complete:
Given a set $X = \{x_1, x_2, \dots , x_{3n}\}$ and a family $\mathcal F = \{F_1, F_2, \dots , F_m\}$,
where each $F_i = \{x_{i_1}, x_{i_2}, x_{i_3}\}$ has three elements of $X$, 
the objective is to find a subfamily $\mathcal F' \subseteq \mathcal F$ such that 
every element in $X$ is contained in exactly one of the subsets $\mathcal F'$. 
By making some copies of $3$-element sets if necessary, 
we may assume that 
$|\{F \in \mathcal F\mid x \in F\} | \ge 3 (n+2)$ for each $x \in X$, 
which implies that $m$ is sufficiently large compared to $n$. 

Given an instance of {\sc Exact 3-cover}
with $|\{F \in \mathcal F\mid x \in F\} | \ge 3 (n+2)$ for each $x \in X$, 
we construct an instance of 
{\sc Maximum Connected Cut} in a split graph as follows. 
We introduce $m$ vertices $u_1, u_2, \dots, u_m$, where each $u_i$ corresponds to $F_i$, 
and introduce $m - 2n$ vertices $u_{m+1}, u_{m+2}, \dots, u_{2(m-n)}$. 
Let $U := \{u_1, u_2, \dots, u_{2(m-n)}\}$. 
For $i=m+1, m+2, \dots, 2(m-n)$, introduce a vertex set $Y_i$ of size $M$, 
where $M$ is a sufficiently large integer compared to $n$ (e.g. $M =3n+1$). 
Now, we construct a graph $G=(U \cup X \cup Y, E)$, where
$Y := \bigcup_{m+1 \le i \le m-2n}Y_i$,
$E_U := \{\{u, u'\} \mid u, u' \in U, u \neq u' \}$, 
$E_X := \{\{u_i, x_j\} \mid  1\le i \le m, 1 \le j \le 3n, x_j \in F_i \}$, 
$E_Y := \{\{u_i, y\} \mid m + 1 \le i \le 2(m-n), y \in Y_i \}$,  and
$E   := E_U \cup E_X \cup E_Y$. 
Then, $G$ is a split graph in which $U$ induces a clique and $X \cup Y$ is an independent set. 
We now show the following claim. 

\begin{claim}
The original instance of {\sc Exact 3-cover} has a solution if and only if 
the obtained graph $G$
has a connected cut of size at least $(m-n)^2 + 3m - 3n + (m-2n) M$. 
\end{claim}

\begin{figure}[tbp]
  \begin{center}
   \includegraphics[width=77mm]{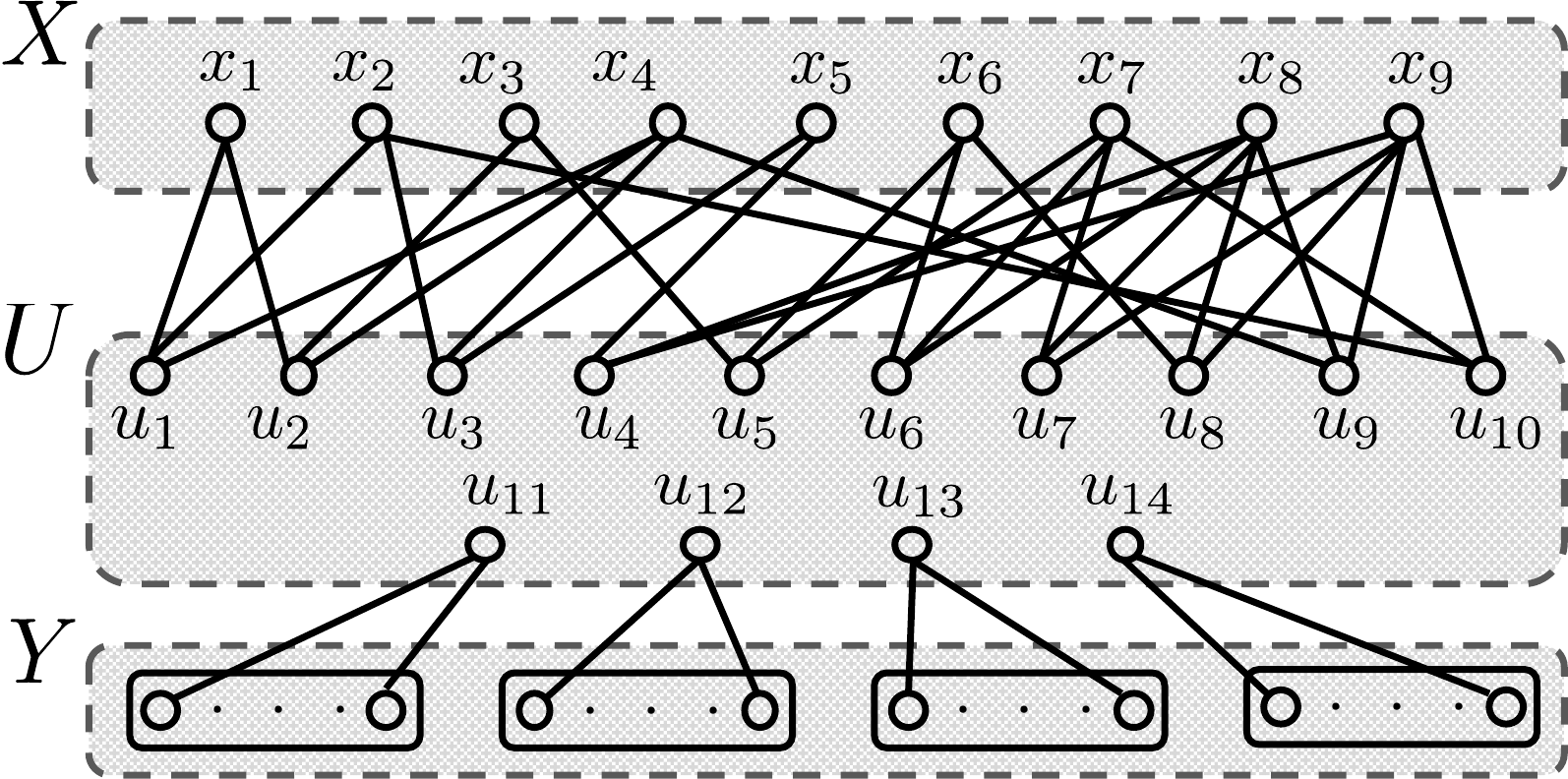}
  \end{center}
  \caption{An instance of \textsc{Maximum Connected Cut} on split graphs reduced
 from an instance of {\sc Exact 3-cover} where $X = \{x_1, x_2,x_3,x_4,x_5, x_6,x_7,x_8,
  x_9\}$ and $\mathcal F = \{\{x_1,x_2,x_3\},\{x_1,x_3,x_4\}$,
  $\{x_2,x_4,x_5\}$,
  $\{x_5,x_8,x_9\}$,
  $\{x_3,x_6,x_7\}$,
  $\{x_6,x_7,x_8\}$,
  $\{x_7,x_8,x_9\}$,
  $\{x_6,x_8,x_9\}$,
  $\{x_4,x_8,x_9\}$,
  $\{x_2,x_7,x_9\}
  \}$.
  }
 \label{CMConsplit}
\end{figure}

\begin{proof}
Suppose that the original instance of {\sc Exact 3-cover} has a solution $\mathcal F'$. 
Then 
$S := \{u_i \mid F_i \in \mathcal F'\} \cup \{u_i \mid m+1\le i \le 2(m-n) \} \cup X$
is a desired connected cut, because 
$|\partial(S) \cap E_U| = (m-n)^2$, 
$|\partial(S) \cap E_X| = \sum_{i=1}^m |F_i| - |X| =  3m-3n$, and
$|\partial(S) \cap E_Y| = (m-2n) M$. 

Conversely, suppose that 
the obtained instance of {\sc Maximum Connected Cut} 
has a connected cut $S$ such that $|\partial(S)| \ge (m-n)^2 + 3m - 3n + (m-2n) M$. 
Since 
$|\partial(S) \cap E_U| \le (m-n)^2$, 
$|\partial(S) \cap E_X| \le  3m$, and
$|\partial(S) \cap E_Y| \le |S \cap \{ u_{m+1}, \dots , u_{2(m-n)} \}| \cdot M$,  
we obtain $|S \cap \{ u_{m+1}, \dots , u_{2(m-n)} \}| = m-2n$, that is, 
$\{ u_{m+1}, \dots , u_{2(m-n)} \} \subseteq S$. 
Let $t = |S \cap \{u_1, \dots ,u_m\}|$, $X_0 = \{x \in X \mid N(x) \cap S = \emptyset \}$ the vertices in $X$ that has no neighbor in $S$, $X_{\rm all} = \{x \in X \mid N(x) \subseteq S\}$ the vertices in $X$ whose neighbor is entirely included in $S$, and $X_{\rm part} = X \setminus (X_0 \cup  X_{\rm all})$ all the other vertices in $X$. 
Recall that every element in $X$ is contained in at least $3(n+2)$ subsets of $\mathcal F$.
Then, since 
$|\partial(S) \cap E_U| = (m - t) (m-2n+t) = (m-n)^2 - (t-n)^2$, 
$|\partial(S) \cap E_X| \le |E_X| -  |X_{\rm part}| - |\partial(X_0)|  \le  3m -  (3n - |X_{\rm all}| - |X_0|) - 3(n+2) |X_0|$, 
$|\partial(S) \cap E_Y| \le (m-2n) M$, and
$|\partial(S)| \ge (m-n)^2 + 3m - 3n + (m-2n) M$,  
we obtain 
\begin{align}\label{eq:01}
|X_{\rm all}| - (3n+5) |X_0| - (t-n)^2 \ge 0. 
\end{align}
By counting the number of edges between $S \cap \{u_1, u_2, \dots, u_m\}$ and $X$, we obtain
$3t \ge |\partial(X_{\rm all})| \ge 3(n+2) |X_{\rm all}|$, 
which shows that 
$t \ge (n+2) |X_{\rm all}|$. 
If $|X_{\rm all}| \ge 1$, then $t \ge (n+2) |X_{\rm all}| \ge n+ 2 |X_{\rm all}|$, and hence
$
|X_{\rm all}| - 3(n+5) |X_0| - (t-n)^2 
\le |X_{\rm all}| - (2 |X_{\rm all}|)^2 <0$, 
which contradicts (\ref{eq:01}). 
Thus, we obtain $|X_{\rm all}| = 0$, and hence 
we have $t=n$ and $X_0 = \emptyset$ by~(\ref{eq:01}). 
Therefore, $\mathcal F':= \{ F_i \mid  1 \le i \in m, \ u_i \in S \}$ 
satisfies that $|\mathcal F'| = n$ and $\bigcup_{F \in \mathcal F'} F = X$. 
This shows that $\mathcal F'$ is a solution of the 
original instance of {\sc Exact 3-cover}.  
\end{proof}

This shows that {\sc Exact 3-cover} is reduced to 
{\sc Maximum Connected Cut} in split graphs, 
which completes the proof.  
\end{proof}

\begin{theorem}\label{thm:split:maxmin}
\textsc{Largest Bond} is NP-complete on split graphs.
\end{theorem}
\begin{proof}
We give a reduction from \textsc{Maximum Cut}.
Given a graph $G=(V,E)$ with $n$ vertices, we create a split graph $G'= (V\cup V_E, E')$  where  $V$ is a clique, $V_E = \{e^{\ell} \mid e \in E, 1 \le \ell \le n^3\}$ is an independent set, and $E'=\{ \{u, e^{\ell}\}, \{v,e^{\ell}\} \mid e=\{u,v\}\in E, 1 \le {\ell} \le n^3\}$. We show that $G$ has a cut of size at least $k$ if and only if $G'$ has a bond of size at least $kn^3$.
Without loss of generality, we assume that $n > 1$ and $k > 2$.

Let $[S_1,S_2]$ be a cut of $G$ of size $k$.
We define a cut $[S'_1, S'_2]$ of $G'$ with $S_i \subseteq S'_i$ for $i \in \{ 1, 2\}$.
For each $e\in E$ and $1 \le {\ell} \le n^3$, we set $e^{\ell} \in S'_2$ if both endpoints of $e$ are in $S_2$ in $G$, and otherwise $e^{\ell}\in S'_1$.
It is straightforward to verify that $G'[S'_1]$ and $G'[S'_2]$ are connected.
If $e = \{u, v\}$ contributes to the cut $[S_1, S_2]$, there are $n^3$ edges ($\{u, e_{\ell}\}$ or $\{v, e_{\ell}\}$) in $G'$ that contribute to $[S'_1, S'_2]$. Therefore, the size of $[S'_1, S'_2]$ is at least $kn^3$.

Conversely, let $[S'_1,S'_2]$ be a bond of size $kn^3$ in $G'$. Let $S_i = S'_i \cap V$ for $i \in \{ 1, 2\}$. For each $e=\{u,w\}$ and $i\in \{1,2\}$, we can observe that $e^{\ell}\in S_i$ if $u,w\in S_i$ due to the connectivity of $S_i$ and $k > 2$.
Since $V$ forms a clique in $G'$, there are at most $n^2$ edges between vertices of $V$ in the cut $[S'_1,S'_2]$.
Thus, at least $kn^3-n^2 > kn^3 - n^3 = (k - 1)n^3$ edges between $V$ and $V_E$ belong to the cutset.
This implies that there are at least $k$ pairs $\{u, w\}$ with $u \in S'_1 \cap V$ and $w \in S'_2 \cap V$, and hence $G$ has a cut of size at least~$k$. 
\end{proof}

\subsection{Inapproximability for Largest Bond}

While the maximum cut of a graph has at least a constant fraction of the edges,
the size of the largest bond can be arbitrarily smaller than the number of
edges; take, e.g., a cycle on $n$ edges, for which a largest bond has size $2$.
This discrepancy is also reflected on the approximability of the problems.
Indeed, we show that {\sc Largest Bond} is strictly harder to approximate than
{\sc Maximum Cut}.
To simplify the presentation, we consider a weighted version of the problem in
which edges are allowed to have weights $0$ or $1$; the hardness results will
follow for the unweighted case as well. In the {\sc Binary Weighted Largest Bond}, the
input is given by a connected weighted graph $H$ with weights $w : E(H) \to
\{0,1\}$. The objective is to find a bond whose total weight is maximum.

Let $G$ be a graph on $n$ vertices and whose maximum cut has size $k$. Next, we
define the $G$-\emph{edge embedding} operator~$\xig$. Given a connected weighted
graph $H$, the weighted graph~$\xig(H)$ is constructed by replacing each edge
$\{u,v\} \in E(H)$ with weight~$1$ by a copy of~$G$, denoted by $G_{uv}$, whose
edges have weight $1$, and, for each vertex~$t$ of $G_{uv}$, new edges~$\{u,t\}$
and~$\{v,t\}$, both with weight $0$.

We can also apply the $G$-edge embedding operation on the graph $\xig(H)$, then
on $\xig(\xig(H))$, and so on. In what follows, for an integer $h \ge 0$, denote
by $\xig^h(H)$ the graph resulting from the operation that receives a graph $H$
and applies $\xig$ successively $h$ times. Notice that $\xig^h(H)$ can be constructed in $\Oh(|V(G)|^{h+1})$ time.
For some $j$, ${0 \le j \le h-1}$, observe that an edge ${\{u,v\} \in
E(\xig^j(H))}$ will be replaced by a series of vertices added in iterations $j +
1, j+ 2, \dots, h$. These vertices will be called the \emph{descendants}
of~$\{u,v\}$, and will be denoted by~$V_{uv}$.

Let $K_2$ be the graph composed of a single edge $\{u,v\}$, and consider the
problem of finding a bond of $\xig(K_2)$ with maximum weight. Since edges
connecting $u$ or $v$ have weight $0$, one can assume that $u$ and $v$ are in
different sides of the bond, and the problem reduces to finding a  maximum cut
of~$G$.
In other words, the operator $\xig$ embeds an instance $G$ of {\sc Maximum Cut}
into an edge~$\{u,v\}$ of~$K_2$.

This suggests the following strategy to solve an instance of {\sc Maximum Cut}.
For some constant integer $h \ge 1$, calculate $H = \xig^h(K_2)$, and obtain a
bond~$F$~of~$H$ with maximum weight.
Note that, to solve $H$, one must solve embedded instances of {\sc Maximum Cut}
in multiple levels simultaneously. For a level~$j$, ${1 \le j \le h - 1}$, each
edge $\{u,v\} \in E(\xig^j(K_2))$ with weight~$1$ will be replaced by a graph
$G_{uv}$ which is isomorphic to $G$.
In Lemma~\ref{lem-edge-cut} below, we argue that $F$ is such that either
$V(G_{uv}) \cup \{u,v\}$ are all in the same side of the cut, or~$u$~and~$v$ are in
distinct sides. In the latter case, the edges of $F$ that separate $u$ and $v$
will induce a cut of $G$.

In the remaining of this section, we consider a constant integer $h \ge 0$.
Then, we define $H^j = \xig^j(K_2)$ for every $j$, ${0 \le j \le h}$, and $H =
H^h$.
Also, we write $[S,T]$ to denote the cut induced by a bond~$F$~of~$H$.

\begin{definition}
Let $F$ be a bond~of~$H$ with cut $[S,T]$.
We say that an edge ${\{u,v\} \in E(H^j)}$ with weight~$1$ is \emph{nice for~$F$}
if either
\begin{itemize}
  \item $|\{u, v\} \cap S| = 1$, or
  \item $(\{u, v\} \cup V_{uv}) \subseteq S$, or
  \item $(\{u, v\} \cup V_{uv}) \subseteq T$.
\end{itemize}
Also, we say that $F$ is \emph{nice} if, for every $j$, ${0 \le j \le
h-1}$, and every edge $\{u,v\} \in E(H^j)$ with weight~$1$, $\{u,v\}$ is nice
for~$F$.
\end{definition}

\begin{lemma}\label{lem-edge-cut}
There is a polynomial-time algorithm that receives a bond $F$, and finds a nice
bond ${F'}$ such that $w(F') = w(F)$.
\end{lemma}

\begin{proof}
Let $[S,T]$ be the cut induced by $F$ and let $j^*$ be the minimum value such that there
exists an edge $\{u,v\} \in H^{j^*}$ with weight~$1$ which is not nice for $F$.
Then ${|\{u, v\} \cap S| \ne 1}$. Assume, without loss of generality that $u, v
\in S$. In this case,  $U := V_{uv} \cap T$ is not empty. Since removing
vertices $\{u, v\}$ disconnects $U$, and $T$ must be connected, it follows that
${U = T}$.
This implies that ${N(T) \subseteq (V_{uv} \setminus T)\cup \{u, v\}}$.

We will construct a bond $F'$ of $H$ with cut $[S', T']$.
Let $S'$ be the set of vertices in the connected component of $H[S \setminus
\{v\}]$ which contains~$u$, and ${T' = V(H) \setminus S'}$.
Since $H[S]$ is connected, so must be $H[S \setminus S']$. Also, each vertex
of~$U$ is adjacent~to~$v$, thus $H[(S \setminus S') \cup U]$ is connected.
Observe that $T' = (S \setminus S') \cup U$, so indeed the cut~$[S', T']$
induces a bond ${F' = \partial(S')}$.
Observe that any edge that appears only in $F$ or only in $F'$ is adjacent
to~$v$. Since such edges have weight~$0$, this implies~$w(F) = w(F')$.

To complete the proof, we claim that if for some $j$, $0 \le j \le h$, there
exists an edge $\{u,v\} \in H^{j}$ with weight~$1$ which is not nice for~$F'$,
then $j > j^*$. If this claim holds, then we need to repeat the previous
procedure at most $h$ times before obtaining a nice bond~$F'$.

To prove the claim, consider an edge $\{s,t\} \in H^{j}$ which is not nice
for~$F'$. Suppose, for a contradiction, that ${V_{st} \cap V_{uv} = \emptyset}$.
There are two possibilities.
If $s, t \in S'$, then $V_{st} \subseteq S'$;
if $s, t \in T'$, then $V_{st} \subseteq S \setminus S' \subseteq T'$.
In either possibility, $\{s,t\}$ is nice for~$F'$. This is a contradiction, and
thus ${V_{uv} \cap V_{st} \ne \emptyset}$.

The statement ${V_{uv} \cap V_{st} \ne \emptyset}$ can only happen if $V_{uv}
\subseteq V_{st}$ or $V_{st} \subseteq V_{uv}$. If $V_{uv} \subseteq V_{st}$,
then $U \subseteq V_{st}$ and $s, t \in S$. This implies that $\{s,t\}$ is not
nice for~$F$. But in this case $j < j^*$, contradicting the choice of~$j^*$.
Therefore, $V_{st} \subseteq V_{uv}$, and $j > j^*$, proving our claim. 
\end{proof}

In the following, assume that $F$ is a nice bond with cut $[S, T]$.
Consider a level $j$, ${0 \le j \le h}$, and an edge $\{u,v\} \in E(H^j)$ with
weight~$1$ such that ${|\{u, v\} \cap S| = 1}$.
If $j < h$, then we define $F_{uv}$ to be the subset of edges in $F$ which are
incident with some vertex of~$V_{uv}$; if $j = h$, then we define $F_{uv} = \{
\{u,v\}\}$.
Note that, because $F$ is nice, if ${|\{u, v\} \cap S| \ne 1}$, then no edge
of~$F$ is incident with $V_{uv}$.

Suppose now that ${|\{u, v\} \cap S| = 1}$ for some edge $\{u,v\} \in E(H^j)$ with
weight~$1$ and $0 \le j \le h - 1$.

In this case, $F$ induces a cut-set of~$G_{uv}$. Namely, define
$
{\hat S_{uv} := S \cap V(G_{uv})}
$
and
$
{\hat T_{uv} := T \cap V(G_{uv})}
$
and let $\hat F_{uv}$ be the cut-set of $G_{uv}$ corresponding to cut $[\hat
S_{uv} ,\hat T_{uv}]$.

Observe that for distinct edges $\{u,v\}$ and $\{s,t\}$, it is possible that ${|\hat
F_{uv}| \ne |\hat F_{st}|}$.
We will consider bonds $F$ for which all induced cut-sets $\hat F_{uv}$ have the
same size.

\begin{definition}
Let $\ell$ be a positive integer.
A bond~$F$ of $H$ with cut $[S, T]$ is said to be \emph{\mbox{$\ell$-uniform}}
if, \emph{(i)} $F$~is~nice, and $(ii)$ for every $j$, ${0 \le j \le h-1}$, and
every edge $\{u,v\} \in E(H^j)$ with weight~$1$ such that ${|\{u,v\} \cap S\}| =
1}$, $|\hat F_{uv}| = \ell$.
\end{definition}

An $\ell$-uniform bond induces a cut-set of $G$ of size $\ell$.

\begin{lemma}\label{lem-induced-cut}
Suppose $F$ is an $\ell$-uniform bond of $H$. One can find in polynomial time a
cut-set $L$ of $G$ with $|L| = \ell$.
\end{lemma}

\begin{proof}
Let $u,v$ be the vertices of $K_2$ to which $\xig$ was applied. Since $F$ is
$\ell$-uniform, $|\hat F_{uv}|= \ell$. Note that $\hat F_{uv}$ induces a cut-set
of size $\ell$ on $G$. 
\end{proof}

In the opposite direction, a cut of $G$ induces an $\ell$-uniform bond~of~$H$.

\begin{lemma}\label{lem-induced-bond}
Suppose $L$ is a cut-set of $G$ with $|L| = \ell$. One can find in polynomial
time an $\ell$-uniform bond $F$ of $H$ with $w(F) = \ell^h$.
\end{lemma}

\begin{proof}
For each $j \ge 0$, we construct a bond~$F^j$~of~$H^j$.
For $j = 0$, let $F^0$ be the set containing the unique edge of $H^0 = K_2$.
Suppose now that we already constructed a bond~$F^{j-1}$ of~$H^{j-1}$. For each
edge $\{u,v\} \in F^{j-1}$, let $L_{uv}$ be the set of edges of $G_{uv}$
corresponding to~$L$. Define $F^j := \cup_{\{u,v\} \in F^{j-1}} L_{uv}$.
One can verify that indeed $F^j$ is a bond of $H^j$, and that $w(F_j) = |L|
\times w(F_{j-1}) = \ell^j$. 
\end{proof}

\begin{lemma} \label{lem-uniform}
There is a polynomial-time algorithm that receives a bond~$F$~of~$H$, and finds
an $\ell$-uniform bond ${F'}$~of~$H$ such that $w(F') = \ell^h \ge w(F)$.
\end{lemma}

\begin{proof}
Let $[S,T]$ be the cut corresponding to $F$. First, find the largest cut-set of
a graph~$G_{uv}$ over cut-sets~$\hat F_{uv}$. More precisely, define $\hat F$ to
be the cut-set $\hat F_{uv}$ with maximum $|\hat F_{uv}|$ over all edges $\{u,v\}
\in E(H^j)$ with weight~$1$ such that ${|\{u,v\} \cap S\}| = 1}$, and over all
$j$, ${0 \le j \le h-1}$.
Let $\ell := |\hat F|$.

We claim that for every $j$, $0 \le j \le h$, and every edge $\{u,v\} \in E(H^j)$
with weight~$1$ such that ${|\{u,v\} \cap S\}| = 1}$, $w(F_{uv}) \le \ell^{h -
j}$.
The proof is by (backward) induction on $j$. For $j = h$, $F_{uv} = \{u,v\}$, so
$w(F_{uv}) = 1$. Next, let $j < h$, and assume the claim holds for $j+1$.

Let $F^0_{uv}$ be the subset of edges in $F_{uv}$ incident with $u$ or $v$.
The set $F_{uv}$ can be partitioned into $F^0_{uv}$ and sets $F_{st}$ for
$\{s,t\} \in \hat F_{uv}$. To see this, observe that each edge $\{x, y\} \in
F_{uv} \setminus F^0_{uv}$ must be incident with descendants of $\{u,v\}$, and
thus $\{x, y\}$ is incident with vertices of~$V_{st}$, for some edge $\{s,t\}
\in E(G_{uv})$. Since $|\{x, y\}\cap S| = 1$, neither $V_{st} \cup \{s, t\}
\subseteq S$, nor $V_{st} \cup \{s, t\} \subseteq T$. Because $F$ is nice, it
follows that $|\{s,t\} \cap S| = 1$, then $\{s,t\} \in \hat F_{uv}$, and thus
$\{x,y\} \in F_{st}$.
To complete the claim, observe that, by the induction hypothesis, ${w(F_{st}) \le
\ell^{h - j - 1}}$ for each $\{s,t\} \in \hat F_{uv}$, and recall that $|\hat
F_{uv}| \le |\hat F|$. Therefore
\[
w(F) = w(F^0_{uv}) + \sum_{\{s,t\} \in \hat F_{uv}} w(F_{st})
     \le |\hat F| \times \ell^{h - j - 1} =  \ell^{h - j}.
\]

Using Lemma~\ref{lem-induced-bond} for $\hat F$, we construct a bond~$F'$
for~$H$ with ${w(F') = \ell^h}$.
\end{proof}

\begin{lemma}\label{lem-opt-sol}
Let $F^*$ be a bond of $H$ with maximum weight. Then $w(F^*) = k^h$.
\end{lemma}

\begin{proof}
We assume that $F^*$ is $\ell$-uniform such that $w(F^*) = \ell^h$ for some
$\ell$; if this is not the case, then use Lemma~\ref{lem-uniform}.

Since $F^*$ is $\ell$-uniform, using Lemma~\ref{lem-induced-bond} one obtains a
cut-set $L$ of $G$ with size~$\ell$, then~$\ell \le k$, and thus $w(F^*) \le
k^h$.

Conversely, let $L$ be a cut-set of $G$ with size $k$. Using
Lemma~\ref{lem-induced-bond} for $L$, we obtain a bond~$F$ of~$H$ with weight
$k^h$, and thus $w(F^*) \ge k^h$. 
\end{proof}

\begin{lemma}\label{lem-weighted-hard}
If there exists a constant-factor approximation algorithm for {\sc Weighted
Largest Bond}, then $\PP = \NP$.
\end{lemma}

\begin{proof}
Consider a graph $G$ whose maximum cut has size~$k$. Construct graph~$H$ and
obtain a bond $F$~of~$H$ using an $\alpha$-approximation, for some constant ${0
< \alpha < 1}$. \mbox{Using} the algorithm of Lemma~\ref{lem-uniform}, obtain an
$\ell$-uniform bond $F'$~of~$H$ such that $w(F') = \ell^h \ge w(F)$.
Using Lemma~\ref{lem-opt-sol} and the fact that $F'$ is an
$\alpha$-approximation, ${\ell^h \ge \alpha \times k^h}$.
Using Lemma~\ref{lem-induced-cut}, one can obtain a cut-set $L$ of $G$ with size
$\ell \ge \alpha^{\frac{1}{h}} k$.

For any constant $\varepsilon$, ${0 < \varepsilon < 1}$, we can set $h =
\lceil\log_{1 - \varepsilon} \alpha \rceil$, such that the cut-set~$L$ has size
at least $\ell \ge (1 - \varepsilon) k$. Since {\sc Maximum Cut} is APX-hard,
this implies $\PP = \NP$. 
\end{proof}

\begin{theorem}
If there exists a constant-factor approximation algorithm for {\sc Largest
Bond}, then $\PP = \NP$.
\end{theorem}

\begin{proof}
We show that if there exists an $\alpha$-approximation algorithm for {\sc
Largest Bond}, for constant $0 < \alpha < 1$, then there is an
$\alpha/2$-approximation algorithm for the {\sc Binary Weighted Largest Bond}, so the
theorem will follow from Lemma~\ref{lem-weighted-hard}.

Let $H$ be a weighted graph whose edge weights are all $0$~or~$1$. Let $m$ be
the number of edges with weight~$0$, and let $l$ be the weight of a bond of
$H$ with maximum weight.
Assume $l\ge 2/\alpha$, as otherwise, one can find an optimal solution in
polynomial time by enumerating sets of up to $2/\alpha$ edges.

Construct an unweighted graph $G$ as follows. Start with a copy of $H$ and, for
each edge $\{u,v\} \in E(H)$ with weight $1$, replace $\{u,v\} \in E(G)$ by $m$
parallel edges. Finally, to obtain a simple graph, subdivide each edge of $G$.
If $F$ is a bond of~$G$, then one can construct a bond~$F'$ of~$H$ by undoing
the subdivision and removing the parallel edges. Each edge of $F'$ has weight
$1$, with exception of at most $m$ edges. Thus,
$
w(F') \ge (|F| - m) / m.
$

Observe that an optimal bond of $H$ induces a bond of $G$ with size at
least~$ml$. Thus, if~$F$ is an $\alpha$-approximation for $G$, then ${|F| \ge
\alpha m l}$ and therefore
\[
w(F') \ge \frac{\alpha m l - m}{m} = \alpha l - 1 \ge \alpha l - \alpha l / 2 = \alpha l / 2.
\]
We conclude that $F'$ is an $\alpha/2$-approximation for $H$. 
\end{proof}

\section{Parameterized algorithms}

In this section, we present parameterized algorithms for {\sc Largest Bond}, {\sc Largest $st$-Bond}, {\sc Maximum Connected Cut}, and {\sc Maximum Connected $st$-Cut}. 

\subsection{Algorithmic upper bounds for clique-width parameterization}

Lemma~\ref{lb_cw} shows that {\sc Largest Bond} on graphs of clique-width $w$ cannot be solved in time $f(w)\times n^{o(w)}$ unless the ETH fails. Now, we show that given an expression tree of width at most $w$, {\sc Largest Bond} can be solved in $f(w)\times n^{O(w)}$ time.


An expression tree $\mathcal{T}$ is irredundant if for any join node $\eta(i,j)$, the vertices labeled
by~$i$ and~$j$ are not adjacent in the graph associated with its child. It was shown by Courcelle and Olariu~\cite{courcelle2000upper} that every expression tree $\mathcal{T}$ of $G$ can be transformed into an irredundant expression tree $\mathcal{T}$ of the same width in time linear in the size of~$\mathcal{T}$. Therefore, without loss of generality, we can assume that $\mathcal{T}$ is irredundant.

Our algorithm is based on dynamic programming over the expression tree of the input graph. We first describe what we store in the tables corresponding to the nodes in the expression tree.

Given a $w$-labeled graph $G$, two connected components of $G$ has the same {\em type} if they have the same set of labels. Thus, a $w$-labeled graph $G$ has at most $2^w-1$ types of connected components.

Now, for every node $X_\ell$ of $\mathcal{T}$, denote by $G_{X_\ell}$ the $w$-labeled graph associated with this node, and let $L_1(X_\ell),\ldots,L_w(X_\ell)$ be the sets of vertices of $G_{X_\ell}$ labeled with $1,\ldots,w$, respectively. 
We define a table where each entry is of the form $$c[\ell,s_{1}, ..., s_{w}, r, e_{1}, ..., e_{2^w-1}, d_{1}, ..., d_{2^w-1}],$$ such that: $0\leq s_i \leq |L_i(X_\ell)|$ for $1\leq i\leq w$; $0\leq r\leq |E(G_{X_\ell})|$; $0\leq e_{i}\leq \min\{2,|L_i(X_\ell)|\}$ for $1\leq i\leq 2^w-1$; and $0\leq d_{i}\leq \min\{2,|L_i(X_\ell)|\}$ for $1\leq i\leq 2^w-1$.

Each entry of the table represents whether there is a partition $V_1,V_2$ of $V(G_{X_\ell})$ such that: $|V_1\cap L_i(G_{X_\ell})| = s_i$; the cut-set of $[V_1,V_2]$ has size at least $r$;  $G_{X_\ell}[V_1]$ has $e_i$ connected components of type $i$; $G_{X_\ell}[V_2]$ has $d_i$ connected components of type $i$, where $e_i=2$ means that $G_{X_\ell}[V_1]$ has {\em at least} two connected components of type $i$. The same holds for $d_i$.

Notice that this table contains $f(w)\times n^{\Oh(w)}$ entries.
If $X_\ell$ is the root node of $\mathcal{T}$ (that is, $G = G_{X_\ell})$, then the size of the largest bond of $G$ is equal to the maximum value of $r$ for which the table for $X_\ell$ contains a valid entry (true value), such that there are $j$ and $k$ such that $e_{i} = 0$, $e_{j} = 1$ for $1 \leq i, j \leq 2^w-1$, $i \neq j$; and $d_{i} = 0$, $d_{k} = 1$ for $1\leq i, k \leq 2^w-1$,~$i \neq k$. Similarly, the maximum connected cut of $G$ can be found in a valid entry for $r$ such that some $e_{j} = 1$, and $e_{i} = 0$ for every $i \neq j$.

It is easy to see that we store enough information to compute a largest bond. Note that a $w$-labeled graph is connected if and only if it has exactly one type of connected components and exactly one component of such a type.

Now we provide the details of how to construct and update such tables. The construction for introduce nodes of $\mathcal{T}$ is straightforward.

\smallskip

\textbf{Relabel node:} Suppose that $X_\ell$ is a relabel node $\rho(i,j)$, and let $X_{\ell'}$ be the child of $X_\ell$. 
The table for $X_\ell$ contains a valid entry $c[\ell,s_{1}, ..., s_{w}, r, e_{1}, ..., e_{2^w-1}, d_{1}, ..., d_{2^w-1}]$ if and only if  the table for $X_{\ell'}$ contains an entry 
$$c[\ell',s'_{1}, ..., s'_{w}, r, e'_{1}, ..., e'_{2^w-1}, d'_{1}, ..., d'_{2^w-1}]=true,$$
where:

$s_i = 0$;

$s_{j}=s'_{i}+s'_{j}$;

$s_{p} = s'_{p}$ for $1\leq p \leq w$, $p\neq i, j$;

$e_p=e'_p$ for any type that contain neither $i$ nor $j$;

$e_p=0$ for any type that contains $i$;

and for any type $e_p$ that contains $j$, it holds that $e_p=\min\{2, e'_p+e'_q+e'_r\}$ where

$e'_q$ represent the set of labels $(C_p\setminus \{j\})\cup \{i\}$,

$e'_r$ represent the set of labels $C_p\cup \{i\}$, and

$C_p$ is the set of labels associated to $p$.
The same holds for $d_{1}, ..., d_{2^w-1}$.

\smallskip

\textbf{Union node:} Suppose that $X_\ell$ is a union node with children $X_{\ell'}$ and $X_{\ell''}$.
It holds that $c[\ell,s_{1}, ..., s_{w}, r, e_{1}, ..., e_{2^w-1}, d_{1}, ..., d_{2^w-1}]$ equals true if and only if there are valid entries 
$$c[\ell', s'_{1}, ..., s'_{w}, r', e'_{1}, ..., e'_{2^w-1}, d'_{1}, ..., d'_{2^w-1}]$$ and $$c[\ell'',s''_{1}, ..., s''_{w}, r'', e''_{1}, ..., e''_{2^w-1}, d''_{1}, ..., d''_{2^w-1}],$$ 
having:

$s_i=s'_{i}+s''_{i}$ for $1 \leq i \leq w$;

$r' + r'' \geq r$;

$e_{k} = \min\{2, e'_{k}+e''_{k}\}$, and

$d_{k} = min\{2, d'_{k} + d''_{k}\}$ for $1\leq k \leq 2^w-1$.

\smallskip

\textbf{Join node:}
Finally, let $X_\ell$ be a join node $\eta(i,j)$ with the child $X_{\ell'}$. Remind that since the expression tree is irredundant then the vertices labeled by $i$ and $j$ are not adjacent in the graph $G_{X_{\ell'}}$. Therefore, the entry $c[\ell,s_{1}, ..., s_{w}, r, e_{1}, ..., e_{2^w-1}, d_{1}, ..., d_{2^w-1}]$ equals true if and only if there is a valid entry 
$$c[\ell',s_{1}, ..., s_{w}, r', e'_{1}, ..., e'_{2^w-1}, d'_{1}, ..., d'_{2^w-1}]$$ 
where
$$r' + s_{i}\times(|L_{j}(X_{\ell'})| - s_{j}) + s_{j}\times(|L_{i}(X_{\ell'})| - s_{i}) \geq r,$$
and
$e_p=e'_p$, case $p$ is associated to a type that contains neither $i$ nor $j$;
$e_{p} = 1$, case $p$ is associated to $C^{\ell'}_{i,j}\setminus \{i\}$, where $C^{\ell'}_{i,j}$ is the set of labels obtained by the union of the types of $G_{X_{\ell'}}$ with some connected component having either label $i$ or label $j$;
$e_{p} = 0$, otherwise.
The same holds for $d_{1}, ..., d_{2^w-1}$.

The correctness of the algorithm follows from the description of the procedure. Since for each $\ell$, there are $\Oh((n+1)^w\times m\times (3^{2^w-1})^2)$ entries, the running time of the algorithm is $f(w)\times n^{\Oh(w)}$. This algorithm  together with Lemma~\ref{lb_cw} concludes the proof of the Theorem~\ref{cliquewidth_up}.

\begin{theorem}\label{cliquewidth_up}
Given an expression tree of width at most $w$, {\sc Largest Bond} and {\sc Maximum Connected Cut} can be solved in time $f(w)\times n^{\Oh(w)}$, but they cannot be solved in time $f(w)\times n^{o(w)}$ unless ETH fails.
\end{theorem}

In order to extend this result to {\sc Largest $st$-Bond} and {\sc Maximum Connected $st$-Cut}, it is enough to observe that given a tree expression $\mathcal{T}$ of $G$ with width $w$, it is easy to construct a tree expression $\mathcal{T}'$ with width equals $w+2$, where no vertex of $V(G)$ has the same label than either $s$ or $t$. Let $w+1$ be the label of $s$, and let $w+2$ be the label of $t$. By fixing, for each $\ell$, $s_{w+1}=|L_{w+1}(X_\ell)|$ and~$s_{w+2}=0$, one can solve {\sc Largest $st$-Bond} in time $f(w)\times n^{\Oh(w)}$.

In addition, we can improve the clique-width dependence for a clique-width parameterization of {\sc Largest Bond} and {\sc Maximum Connected Cut}, using another parameter, called "{\em module width}", proposed by \cite{Rao2008}. However, since the algorithms are complicated, for simplicity and readability, we present them only in Appendix.

\begin{theorem}\label{cliquewidth_up_faster}
Given an expression tree of width at most $w$, {\sc Largest Bond} and {\sc Maximum Connected Cut} can be solved in time $n^{\Oh(w)}$.
\end{theorem}

\subsection{Bounding the treewidth of G to a linear function of $k$}

Now, we deal with the problems parameterized by the size of the solution ($k$).
We consider the strategy of preprocessing the input in order to bound the treewidth of the resulting instance. We start our analysis with {\sc Largest Bond}.

\begin{definition}
A graph $H$ is called a minor of a graph $G$ if $H$ can be formed from $G$ by deleting edges, deleting vertices, and by contracting edges. For each vertex $v$ of $H$, the set of vertices of $G$ that are contracted into $v$ is called a branch set of $H$.
\end{definition}

\begin{lemma}\label{k2kimpliesbond}
Let $G$ be a simple connected undirected graph, and $k$ be a positive integer. If $G$ contains $K_{2,k}$ as a minor then $G$ has a bond of size at least $k$.
\end{lemma}
\begin{proof}
Let $H$ be a minor of $G$ isomorphic to $K_{2,k}$. Since $G$ is connected and each branch set of $H$ induces a connected subgraph of $G$, from $H$ it is easy to construct a bond of $G$ of size at least $k$.
\end{proof}

Combined with Lemma~\ref{k2kimpliesbond}, the following results show that, without loss of generality, our study on $k$-bonds can be reduced to graphs of treewidth $\Oh(k)$.

\begin{lemma}\cite{bodlaender1997interval}
Every graph $G=(V, E)$ contains $K_{2,k}$ as a minor or has treewidth at most $2k-2$.
\end{lemma}

\begin{lemma}\cite{bodlaender1997interval}\label{algoK2k}
There is a polynomial-time algorithm that either concludes that the input graph $G$ contains $K_{2,k}$ as a minor, or outputs a tree-decomposition of $G$ of width at most $2k-2$.
\end{lemma}

From Lemma~\ref{k2kimpliesbond} and Lemma~\ref{algoK2k}, the following holds.

\begin{corollary}\label{cor:bounding_tw}
There is a polynomial-time algorithm that either concludes that the input graph $G$ contains a bond of size $k$, or outputs a tree-decomposition of $G$ of width at most $2k-2$.
\end{corollary}

Since $k$-bonds are also connected cuts, it holds that there is a polynomial-time algorithm that either concludes that the input graph $G$ contains a bond of size $k$, or outputs a tree-decomposition of $G$ of width at most $2k-2$. Such a bound can be improved to $k-1$ by replace $K_{2,k}$ with $K_{1,k}$ (see~\cite{bodlaender1997interval}) as a minor for {\sc Maximum Connected Cut}.


\subsubsection{The st-bond case}

Let $S\subseteq V(G)$ and let $\partial(S)$ be a bond of a connected graph $G$. Recall
that a block is a $2$-vertex-connected subgraph of $G$ which is inclusion-wise
maximal, and a block-cut tree of $G$ is a tree whose vertices represent the blocks and the cut vertices of $G$, and there is an edge in the block-cut tree for each pair of a block and a cut vertex that belongs to that block.
Then, $\partial(S)$ intersects at most one block of~$G$. More
precisely, for any two distinct blocks $B_1$ and $B_2$ of~$G$, if $S \cap V(B_1)
\neq \emptyset$ and $S \cap V(B_1) \neq V(B_1)$, then either $V(B_2) \subseteq
S$, or $V(B_2) \subseteq V \setminus S$. Indeed, if this is not the case, then
either $G[S]$ or $G[V\setminus S]$ would be disconnected.
Thus, to solve {\sc Largest $st$-Bond}, it is enough to consider, individually, each block on the path between $s$ and $t$ in the block-cut tree of $G$. Also, if a block is composed of a single edge, then it is a bridge~in~$G$, which is not a solution for the problem unless $k = 1$.
Thus, we may assume without loss of generality that $G$ is
$2$-vertex-connected.

\begin{lemma} \label{twopaths}
Let $G$ be a $2$-vertex-connected graph. For all $v \in V(G)\setminus\{s,t\}$, there is an $sv$-path and a $tv$-path which are internally disjoint.
\end{lemma}

\begin{proof}
Since $G$ is $2$-vertex-connected, there are two disjoint $sv$-paths $P_s$ and $P'_s$ and there is a $tv$-path $P_t'$ which does not include $s$. Let $x$ be the first vertex of $P_t'$ which belongs to $V(P_s \cup P_s')$ and assume, w.l.o.g., that $x \in P_s'$. Let $P_t''$ be the sub-path of $P_t'$ from $t$ to $x$ and $P_s''$ the sub-path of $P_s'$ from $x$ to $v$. Now define $P_t$ as $tP_t''xP_s''v$ and notice that $P_t$ is a $tv$-path disjoint from $P_s$. 
\end{proof}


\begin{lemma}\label{st-minor}
Let $G$ be a $2$-vertex-connected graph. If $G$ contains $K_{2,2k}$ as a minor,
then there exists $S \subseteq V(G)$ such that $\partial(S)$ is a bond of size
at least $k$.
\end{lemma}

\begin{proof}
Let $G$ be a graph containing a $K_{2,2k}$ as a minor. If $k = 1$, the
statement holds trivially, thus assume $k \ge 2$. Also, since $G$ is connected,
one can assume that this minor was obtained by contracting or removing edges
only, and thus its branch sets contain all vertices of~$G$.
Let $A$ and $B$ be the branch sets corresponding to first side of $K_{2,2k}$,
and let $X_1, X_2, \dots, X_{2k}$ be the remaining branch sets.

First, suppose that $s$ and $t$ are in distinct branch sets. If this is the
case, then there exist distinct indices $a, b \in \{1, \dots, 2k\}$ such that $s
\in A \cup X_a$ and $t \in B \cup X_b$. Now observe that $G[A \cup X_a]$ and
$G[B \cup X_b]$ are connected, which implies an $st$-bond with at least $2k - 1 \ge
k$ edges. Now, suppose that $s$ and $t$ are in the same branch set. In this
case, one can assume without loss of generality that $s, t \in A \cup X_{2k}$.

Define $U = A \cup X_{2k}$ and $Q = V(G) \setminus U$. Observe that $G[U]$ and
$G[Q]$ are connected. Consider an arbitrary vertex $v$ in the set $Q$. Since $G$
is $2$-vertex-connected, Lemma~\ref{twopaths} implies that there exist an
$sv$-path $P_s$ and a $tv$-path $P_t$ which are internally disjoint. Let $P'_s$
and $P'_t$ be maximal prefixes of $P_s$ and $P_t$, respectively, whose vertices
are contained in $U$.

We partition the set $U$ into parts $U_s$ and $U_t$ such that $G[U_s]$ and
$G[U_t]$ are connected.
Since $G[U]$ is connected, there exists a tree $T$ spanning $U$.
Direct all edges of $T$ towards $s$ and
partition $U$ as follows.
Every vertex in $P'_s$ belongs to $U_s$ and every vertex in $P'_t$  belongs to
$U_t$. For a vertex $u \notin V(P'_s \cup P'_t)$, let $w$ be the first
ancestor of $u$ (accordingly to $T$) which is in $P'_s \cup P'_t$. Notice that
$w$ is well-defined since $u \in V(T)$ and the root of $T$ is $s \in V(P'_s \cup
P'_t)$.
Then $u$ belongs to $U_s$ if $w \in V(P'_s)$, and $u$ belongs to $U_t$ if $w \in
V(P'_t)$.

Observe that that there are at least $2k-1$ edges between $U$ and $Q$, and thus
there are at least $k$ edges between $U_s$ and $Q$, or between $U_t$ and $Q$.
Assume the former holds, as the other case is analogous.
It follows that $G[U_s]$ and $G[U_t \cup Q]$ are connected and induce a bond
of~$G$ with at least $k$ edges. 
\end{proof}

Lemma~\ref{algoK2k} and Lemma~\ref{st-minor} imply that there is an algorithm
that either concludes that the input graph $G$ has a bond of size at least $k$,
or outputs a tree-decomposition of an equivalent instance $G'$ of width
$\Oh(k)$.

\begin{corollary}\label{st_bounding_treewidth}
Given a graph $G$, vertices $s,t\in V(G)$, and an integer $k$, there exists a polynomial-time algorithm that either concludes that $G$ has an $st$-bond of size at least $k$ or outputs a subgraph $G'$ of $G$ together with a tree decomposition of $G'$ of width equals $\Oh(k)$, such that $G'$ has an $st$-bond of size at least $k$ if and only if $G$ has an $st$-bond of size at least~$k$.
\end{corollary}

\begin{proof}
Find a block-cut tree of $G$ in linear time~\cite{CormenLRS01}, and let $B_s$
and $B_t$ be the blocks of $G$ that contain $s$ and $t$, respectively. Remove
each block that is not in the path from $B_s$ to $B_t$ in the block-cut tree
of~$G$. Let $G'$ be the remaining graph.
For each block $B$ of $G'$, consider the vertices $s'$ and $t'$ of $B$ which are
nearest to $s$ and $t$, respectively. Using
Lemmas~\ref{algoK2k}~and~\ref{st-minor} one can in polynomial time either
conclude that $B$ has an \mbox{$s't'$-bond}, in which case $G$ is a
yes-instance, or compute a tree decomposition of $B$ with width at most
$\Oh(k)$.

Now, construct a tree decomposition of $G'$ as follows. Start with the union of
the tree decompositions of all blocks of $G'$. Next, create a bag $\{u\}$ for
each cut vertex $u$ of $G'$.
Finally, for each cut vertex $u$ and any bag corresponding to a block $B$
connected through $u$, add an edge between~$\{u\}$ and one bag of the tree
decomposition of $B$ containing $u$.
Note that this defines a tree decomposition of $G'$ and that each bag has at
most $\Oh(k)$ vertices. 
\end{proof}


Since $st$-bonds are solutions for {\sc Maximum Connected $st$-Cut}, the results presented in Corollary~\ref{st_bounding_treewidth} naturally apply to such a problem as well.

\subsection{Taking the treewidth as parameter}

In the following, given a tree decomposition $\mathcal{T}$, we denote by $\ell$ one node of~$\mathcal{T}$ and by $X_\ell$ the vertices contained in the \emph{bag} of $\ell$. We assume w.l.o.g that $\mathcal{T}$ is a extended version of a \emph{nice} tree decomposition (see~\cite{cygan2015parameterized}), that is, we assume that there is a special root node~$r$ such that $X_\ell = \emptyset$ and all edges of the tree are directed towards $r$ and each node~$\ell$ has one of the following five types: {\em Leaf}\,; {\em Introduce vertex}; {\em Introduce edge}; {\em Forget vertex}; and {\em Join}.
%
%
%
%
%
%
Moreover, define $G_\ell$ to be the subgraph of $G$ which contains only vertices and edges that have been introduced in $\ell$ or in a descendant of $\ell$.

The number of partitions of a set of $k$ elements is the \textit{$k$-th Bell number}, which we denote by $B(k)$ ($B(k)\leq k!$~\cite{bell}).

\begin{theorem}\label{dynamic}
Given a nice tree decomposition of~$G$ with width $tw$, one can find a bond of maximum size in time
$2^{\Oh(tw\log{tw})}\times n$
where $n$ is the
number of vertices of~$G$.
\end{theorem}

\begin{proof}
Let $\partial_{G}(U)$ be a bond of $G$, and $[U, V\setminus U]$ be the cut defined by such a bond.
Set $S^\ell_U=U\cap X_\ell$. The removal of $\partial_{G}(U)$ partitions $G_\ell[U]$ into a set $C^\ell_{U}$ of connected components, and $G_\ell[V\setminus U]$ into a set $C^\ell_{V\setminus U}$ of connected components. Note that~$C^\ell_{U}$ and~$C^\ell_{V\setminus U}$ define partitions of $S^\ell_U$ and $X_\ell\setminus S^\ell_U$, denoted by $\rho^\ell_1$~and $\rho^\ell_2$ respectively, where the intersection of each connected component of $C^\ell_{U}$ with $S^\ell_U$ corresponds to one part of $\rho^\ell_1$. The same holds for~$C^\ell_{V\setminus U}$ with respect to $X_\ell \setminus S^\ell_U$ and $\rho^\ell_2$.

We define a table for which an entry $c[\ell, S, \rho_1, \rho_2]$ is the size of a largest cut-set (partial solution) of the subgraph~$G_\ell$, where~$S$ is the subset of $X_\ell$ to the left part of the bond, $ X_\ell \setminus S$ is the subset to the right part, and
$\rho_1$, $\rho_2$ are the partitions of $S$ and  $X_\ell \setminus S$ representing, after the removal of the partial solution, the intersection with the connected components to the left and to the right, respectively.
If there is no such a partial solution then $c[\ell, S, \rho_1, \rho_2] = -\infty$.

For the case that $S$ is empty, two special cases may occur: either $U\cap V(G_\ell) = \emptyset$,  in which case there are no connected components in $C^\ell_U$, and thus ${\rho_1 = \emptyset}$; or $C^\ell_U$ has only one connected component which does not intersect $X_\ell$, i.e., $\rho_1=\{\emptyset\}$, this case means that the connected component in $C^\ell_U$ was completely forgotten.
Analogously, we may have $\rho_2 = \emptyset$ and $\rho_2=\{\emptyset\}$.
%
%
Note that we do not need to consider the case $\{\emptyset\}\subsetneq \rho_i$ since it would imply in a disconnected solution. The largest bond of a connected graph $G$ corresponds to the root entry $c[r, \emptyset, \{\emptyset\}, \{\emptyset\}]$.

%
To describe a dynamic programming algorithm, we only need to present the recurrence relation for each node type.


\medskip

\textbf{Leaf:} In this case, $X_\ell = \emptyset$. There are a few combinations for $\rho_1$ and $\rho_2$: either $\rho_1 = \emptyset$, or $\rho_1 = \{\emptyset\}$, and either $\rho_2 = \emptyset$, or $\rho_2 = \{\emptyset\}$.
Since for this case  $G_\ell$ is empty, there can be no connected components, so having $\rho_1 = \emptyset$ and $\rho_2 = \emptyset$ is the only feasible choice.
\[
c[\ell, S, \rho_1, \rho_2] =
\begin{cases}
0  & \text{if $S = \emptyset$, $\rho_1 = \emptyset$ and $\rho_2 = \emptyset$},\\
-\infty & \text{otherwise}.
\end{cases}
\]

\smallskip

\textbf{Introduce vertex:} 
We have only two possibilities in this case, either $v$ is an isolated vertex to the left ($v \in S$) or it is an isolated vertex to the right ($v \notin S$). Thus, a partial solution on $\ell$ induces a partial solution on $\ell'$, excluding $v$ from its part. 
\[
c[\ell, S, \rho_1, \rho_2] =
\begin{cases}
c[\ell', S\setminus\{v\}, \rho_1\setminus\{\{v\}\}, \rho_2] & \text{if $\{v\} \in \rho_1$}, \\
c[\ell', S, \rho_1, \rho_2\setminus\{\{v\}\}] & \text{if $\{v\} \in \rho_2$}, \\
- \infty & \text{if $\{v\} \notin \rho_1 \cup \rho_2$}. \\
\end{cases}
\]

\textbf{Introduce edge:} In this case, either the edge $\{u,v\}$ that is being inserted is incident with one vertex of each side, or the two endpoints are at the same side. In the former case, a solution on $\ell$ corresponds to a solution on $\ell'$ with the same partitions, but with value increased. In the latter case, edge~$\{u,v\}$ may connect two connected components of a partial solution on $\ell'$.
\[
c[\ell, S, \rho_1, \rho_2] =
\begin{cases}
c[\ell', S, \rho_1, \rho_2] + 1 & \text{if $u \in S$ and $v \notin S$ or $u \notin S$ and $v \in S$} ,\\
max_{ \rho_1'}\{c[\ell', S, \rho_1', \rho_2]\} & \text{if $u \in S$ and $v \in S$}, \\
max_{ \rho_2'}\{c[\ell', S, \rho_1, \rho_2']\} & \text{if $u \notin S$ and $v \notin S$}.
\end{cases}
\]
Here, $\rho_1'$ spans over all refinements of $\rho_1$ such that the union of the parts containing $u$ and $v$ results in the partition $\rho_1$. The same holds for $\rho_2'$.


\smallskip

\textbf{Forget vertex:} In this case, either the forgotten vertex $v$ is in the left side of the partial solution induced on $\ell$, or is in the right side. Thus, $v$ must be in the connected component which contains some part of $\rho_1$, or some part of $\rho_2$. We select the possibility that maximizes the value
\[
c[\ell, S, \rho_1, \rho_2] = max_{ \rho_1',\rho_2'}\{ c[\ell', S \cup \{v\}, \rho_1', \rho_2], c[\ell', S, \rho_1, \rho_2'] \}.
\]
Here, $\rho_1'$ spans over all partitions obtained from~$\rho_1$ by adding $v$ in some part of~$\rho_1$ (if $\rho_1=\{\emptyset\}$ then $\rho_1'=\{v\}$). The same holds for $\rho_2'$.


\smallskip

\textbf{Join:} This node represents the join of $G_{\ell'}$ and $G_{\ell''}$ where ${X_\ell = X_{\ell'} = X_{\ell''}}$.

By counting the bond edges contained in $G_{\ell'}$ and in
$G_{\ell''}$, each edge is counted at least once, but edges in $X_{\ell}$ are
counted twice. Thus
\[
c[\ell, S, \rho_1, \rho_2] = max\{c[\ell', S, \rho_1', \rho_2'] + c[\ell', S, \rho_1'', \rho_2'']\} - |\{\{u,v\} \in E, u \in S, v \in X_{\ell} \setminus S\}|.
\]
In this case, we must find the best combination between the two children. Namely, for $i\in\{1,2\}$, we consider combinations of $\rho_i'$ with $\rho_i''$ which merge into $\rho_i$. If $\rho_i=\{\emptyset\}$ then either $\rho_i'=\{\emptyset\}$ and $\rho_i''=\emptyset$; or $\rho_i'=\emptyset$ and $\rho_i''=\{\emptyset\}$. Also, if $\rho_i=\emptyset$ then $\rho_i'=\emptyset$ and $\rho_i''=\emptyset$.

\bigskip

The running time of the dynamic programming algorithm can be estimated as follows. The number of nodes in the decomposition is $\Oh(tw\times n)$~\cite{cygan2015parameterized}.
For each node~$\ell$, the parameters $\rho_1$ and $\rho_2$ induce a partition of $X_\ell$; the number of partitions of $X_\ell$ is given by the corresponding Bell number, $B(|X_\ell|) \le B(tw+1)$.
Each such a partition $\rho$ corresponds to a number of choice of parameter $S$ that corresponds to a subset of the parts of $\rho$; thus the number of choices for $S$ is not larger than ${2^{|\rho|} \le 2^{|X_\ell|} \le 2^{tw+1}}$.
Therefore, we conclude that the table size is at most $\Oh(B(tw+1) \times 2^{tw} \times tw \times n)$.
Since each entry can be computed in $2^{\Oh(tw\log tw)}$ time, the total complexity is $2^{\Oh(tw\log{tw})}\times n$. The correctness of the recursive formulas is straightforward. 
\end{proof}

The reason for the $2^{\Oh(tw \log tw)}$ dependence on treewidth is because we
enumerate all partitions of a bag to check connectivity.
However, one can obtain single exponential-time dependence by modifying the
presented algorithm using rank-based approach as described in Section~\ref{rank}.

\begin{theorem}\label{th:st_tw_log_tw}
{\sc Largest $st$-Bond}, {\sc Maximum Connected Cut}, and {\sc Maximum Connected $st$-Cut} can be solved in time $2^{\Oh(tw\log{tw})} \times n$.
\end{theorem}
\begin{proof}
The solution of {\sc Largest $st$-Bond} can be found by a dynamic programming as presented in Theorem~\ref{dynamic} where we add $s$ and $t$ in all the nodes and we fix~$s\in S$ and~$t\notin S$. To obtain similar algorithms for \textsc{Maximum Connected Cut} and \textsc{Maximum Connected $st$-Cut} we just do not have to take care with the connectivity information for $S_2$ and simply drop it in the above computation. 
\end{proof}

The dynamic programming algorithms in Theorems~\ref{dynamic} and ~\ref{th:st_tw_log_tw} can be seen as ones for {\em connectivity problems} such as finding a Hamiltonian cycle, a feedback vertex set, and a Steiner tree. For such problems, we can improve the running time $2^{O(tw\log tw)}$ to $2^{O(tw)}$ using two techniques called the {\em rank-based approach} due to Bodlaender et al.~\cite{Bodlaender2015} and the {\em cut \& count technique} due to Cygan et al.~\cite{Cygan2011}. In the next two subsections, we improve the running time of the algorithms described in this section using these techniques.




\subsubsection{Rank-based approach}\label{rank}

In this subsection, we provide  faster $2^{O(tw)}$-time deterministic algorithms parameterized by tree-width. To show this, we use the rank-based approach proposed by Bodlaender et al. \cite{Bodlaender2015}. 
The key idea of the rank-based approach is to keep track of {\em small} representative sets of size $2^{O(tw)}$ that capture partial solutions of an optimal solution instead of $2^{O(tw \log tw)}$ partitions.
Indeed, we can compute small representative sets within the claimed running time using {\sf reduce} algorithm \cite{Bodlaender2015}.

We begin with some definition used in the Rank-based approach.
\begin{definition}[Set of weighted partitions \cite{Bodlaender2015}]
Let $\Pi(U)$ be the set of all partitions of some set $U$. A set of weighted partitions is ${\mathcal A} \subseteq \Pi(U)\times {\mathbb N}$, i.e., a family of pairs, each consisting of a partition of $U$ and a non-negative integer weight. 
\end{definition}

The weight of a partition corresponds to the size of a partial solution.
For $p,q\in \Pi(U)$, let $J(p, q)$ denote the join of the partition. We say that a set of weighted partitions ${\mathcal A}' \subseteq \Pi(U)\times {\mathbb N}$ {\em represents} another set ${\mathcal A} \subseteq \Pi(U)\times {\mathbb N}$, if for all $q\in \Pi(U)$ it holds that  $\max\{w\mid (p,w)\in {\mathcal A}' \land J(p, q) = \{U\}\}=\max\{w\mid (p,w)\in {\mathcal A} \land J(p, q)=\{U\}\}$.
Then Bodlaender et al. \cite{Bodlaender2015} provided {\sf reduce} algorithm that computes a {\em small} representative set of weighted partitions.

\begin{theorem}[{\sf reduce} \cite{Bodlaender2015}]\label{thm:tw:reduce}
There exists an algorithm {\sf reduce} that given a set of weigh\-ted partitions ${\mathcal A} \subseteq \Pi(U)\times {\mathbb N}$, outputs in time $|{\mathcal A}|2^{(\omega-1)|U|}|U|^{O(1)}$ a set of weigh\-ted partitions ${\mathcal A}'\subseteq {\mathcal A}$ such that ${\mathcal A}'$ represents ${\mathcal A}$ and $|{\mathcal A'}| \le 2^{|U|-1}$, where $\omega < 2.3727$ denotes the matrix multiplication exponent.
\end{theorem}

The {\sf reduce} algorithm allows us to compute an optimal solution without keeping all weighted partitions.
We apply {\sf reduce} algorithm to the set of partitions at each node in the $O^*({2}^{O(tw\log tw)})$-time algorithm of the previous section.

\begin{theorem}\label{thm:rank-based:concut}
Given a tree decomposition of width $tw$, there are $O^*((1+2^{\omega+1})^{tw})$-time deterministic algorithms for \textsc{Maximum Connected Cut} and \textsc{Maximum Connected $st$-Cut}.
\end{theorem}
\begin{proof}
For a bag $X_\ell$, we compute the value $c[\ell,S_i, \rho_1]$ for each $S_i\subseteq X_\ell$ and a partition $\rho_1$ of $S_i$.
For each $S_i$ and $T_i$,
we apply the {\sf reduce} algorithm to a set of weighted partitions $(\rho_1, c[\ell, S_i, \rho_1])$ that are obtained by recursive formulas as described in the previous section.
At each node $i$, the {\sf reduce} algorithm outputs only $2^{|S_i| - 1}$ weighted partitions for each $S_i$.
Thus, at each node except join nodes, the running time of evaluating the recursive formula is $\sum_{S_i \subseteq X_i}O^*(2^{|S_i|}) = O^*(3^{tw})$ and of the {\sf reduce} algorithm is $\sum_{S_i \subseteq X_i}O^*(2^{|S_i|}2^{(\omega - 1)|S_i|}) = \sum_{S_i \subseteq X_i}O^*(2^{\omega|S_i|}) = O^*((1 + 2^{\omega})^tw)$.
At each join node, since the output of evaluating the recursive formula may contain $O^*(2^{2|S_i|})$ weighted partitions for each $S_i$. 
Thus, the total running time at join node $i$ is $$\sum_{S_i \subseteq X_i}O^*(2^{2|S_i|}2^{(\omega - 1)|S_i|}) = O^*((1 + 2^{\omega + 1})^{tw}).$$
Hence, the theorem follows. 
\end{proof}

Note that if a tree decomposition has no join nodes, namely a path decomposition, the overall running time is $O^*((1 + 2^{\omega})^pw)$.

\begin{theorem}\label{thm:rank-based:minimalcut}
Given a tree decomposition of width $tw$, there are $O^*(2^{(\omega +2)tw})$-time deterministic algorithms for \textsc{Largest Bond} and \textsc{Largest $st$-Bond}.
\end{theorem}
\begin{proof}
For a bag $X_\ell$, we compute the value $c[\ell, S_i, \rho_1, \rho_2]$ for each $S_i \subseteq X_\ell$ and $\rho_1, \rho_2$ being partitions of $S_i$ and $X_\ell\setminus S_i$, respectively.
Similar to Theorem~\ref{thm:rank-based:concut}, for each $S_i$,
we apply the {\sf reduce} algorithm to a set of weighted partitions $(\rho_1, c[\ell, S_i, \rho_1, \rho_2])$ and then apply it again to weighted partitions $(\rho_2, c[\ell,S_i, T_i, \rho_1', \rho_2])$ for each $S_i, T_i$, and for each remaining $\rho_1$ of the first application.
Since there are at most $2^{|S_i|-1}2^{|T_i|-1} = 2^{|X_i|-2}$ weighted partitions in the representative set for each $S_i \subseteq X_i$,
the total running time is $\sum_{S_i \subseteq X_i}O^*(2^{2(|X_i|-2)}2^{(\omega-1)|X_i|}) = O^*(2^{(\omega + 2)tw})$. 
\end{proof}

\begin{corollary}\label{thm:singleexpo}
\textsc{Largest Bond}, \textsc{Largest $st$-Bond}, \textsc{Maximum Connected Cut} and \textsc{Maximum Connected $st$-Cut} can be solved in $2^{O(k)}\times n^{O(1)}$ time.
\end{corollary}
\begin{proof}
It follows directed form Corollary~\ref{cor:bounding_tw}, Corollary~\ref{st_bounding_treewidth}, Theorem~\ref{thm:rank-based:concut} and Theorem~\ref{thm:rank-based:minimalcut}. 
\end{proof}

\subsubsection{{Cut \& Count}}

In this subsection, we design much faster randomized algorithms by using Cut \& Count, which is the framework for solving the connectivity problems faster~\cite{Cygan2011}. 
In Cut \& Count,  we count the number of {\em relaxed} solutions modulo 2 on a tree decomposition and determine whether there exists a connected solution by cancellation tricks.

\begin{definition}[\cite{Cygan2011}]
A cut $[V_1,  V_2]$ of $V' \subseteq V$ such that $V_1\cup V_2 = V'$ and $V_1\cap V_2 = \emptyset$ is {\em consistent} if $v_1 \in V_1$ and $v_2 \in V_2$ implies $(v_1, v_2) \notin E$. 
\end{definition}

In other words, a cut $[V_1, V_2]$ of $V'$ is consistent if there are no edge between $V_1$ and $V_2$. 
 
Fix an arbitrary vertex $v$ in $V_1$. If $G[V]$ has $k$ components, then there exist $2^{k-1}$ consistent cuts of $V$. Thus, when $G[V]$ is connected, there only exists one consistent cut $[V_1, V_2]=[V, \emptyset]$. From this observation, $G[V]$ is connected if and only if the number of consistent cuts is odd. 
Therefore, in order to compute ``connected solutions'', it seems to suffice to count the number of consistent cuts modulo two at first glance.
However, this computation may fail to count the number of ``connected solutions'' since there can be even number of such solutions. 
To overcome this obstacle, Cygan et al.~\cite{Cygan2011} used the Isolation Lemma~\cite{MVV1987}, which ensures with high probability that the problem has a unique minimum solution.
For the detail of the Isolation Lemma, see~\cite{cygan2015parameterized,MVV1987}.



We follow the Cut \& Count framework in \cite{cygan2015parameterized,Cygan2011}: We apply it to determining whether there exists a minimal $st$-cut, a cut that separates $s$ and $t$, of size $k$, namely $st$-bond.
Recall that $[S, V \setminus S]$ is a $st$-bond of a connected graph $G = (V, E)$ if both $G[S]$ and $G[V\setminus S]$ are connected, $s \in S$, and $t \in V \setminus S$.

Let $i$ be the index of a node of a nice tree decomposition of $T$ of $G$.

\begin{definition}
Let $r \in \{1, 2, \ldots, |E|\}$. Let $S^l_i, S^r_i, T^l_i, T^r_i$ be pairwise disjoint (possibly) subsets of $X_i$ such that $S^l_i \cup S^r_i \cup T^l_i \cup T^r_i = X_i$.
A \emph{partial solution} for $(S^l_i, S^r_i, T^l_i, T^r_i, r)$ is a cut $[S, V_i \setminus S]$ of $G_i$ such that:
\begin{itemize}
	\item $S \cap X_i = S^l_i \cup S^r_i$ and $(V_i \setminus S) \cap X_i = T^l_i \cup T^r_i$,
	\item $[S_i^l,S_i^r]$ and $[T^l_i, T^r_i]$ are consistent cuts of $S^l_i \cup S^r_i$ and $T^l_i \cup T^r_i$, respectively, 
	\item there are exactly $r$ cut edges between $S$ and $V_i \setminus S$ in $G_i$, and 
	\item $s\in V_i \implies s\in S_i^l$ and $t\in V_i \implies t\in T_i^l$.
\end{itemize}
\end{definition}

Before proceeding to our dynamic programming, we assign a weight $w_v$ to each vertex $v \in V$ by choosing an integer from $\{1, \ldots, 2n\}$  independently and uniformly at random.
We also use the following preprocessing: add $s$ and $t$ to each node of $T$ and remove the bags introduce $s$ or $t$ from $T$.
In our dynamic programming algorithm, for each node $i$ and for $0 \le w \le 2n^2$ and $0\le r \le |E|$, 
we count the number of partial solutions $[S, V_i\setminus S]$ for $(S^l_i, S^r_i, T^l_i, T^r_i, r)$ such that the total weight of $S$ is exactly $w$, which we denote by $c[i,S^l_i, S^r_i, T^l_i, T^r_i, r, w]$.
By the Isolation Lemma, with high probability, there is a minimal $st$-cut of $G$ of size exactly $k$ if and only if $c(i,\{s\}, \emptyset, \{t\}, \emptyset, k, w)$ is odd 
for some $0 \le w \le 2n^2$ in the root node $r(T)=i$.
In the following, we describe the recursive formula for our dynamic programming.

\paragraph*{Leaf node:}  In a leaf node $i$, since $X_i= \{s, t\}$, we have $c[i,S^l_i, S^r_i, T^l_i, T^r_i, r, w]=1$  if $S_i^l=\{s\}$, $T^l_i=\{t\}$, $S_{1}^r=T^r_i=\emptyset$, $r=0$, and $w=0$. 
Otherwise, $c[i,S^l_i, S^r_i, T^l_i, T^r_i, r, w]=0$.

\medskip

\paragraph*{Introduce vertex  $v$ node:} 
In an introduce vertex node $i$, we consider the following four cases:
\begin{eqnarray*}
c[i,S^l_i, S^r_i, T^l_i, T^r_i, r, w]
& = \begin{cases}
c[j,S^l_i\setminus \{v\}, S^r_i, T^l_i, T^r_i, r, w-w(v)]& \mbox{if  $v\in S^l_i$,} \\
c[j,S^l_i, S^r_i\setminus \{v\},, T^l_i, T^r_i, r, w-w(v)]& \mbox{if  $v\in S^r_i$,} \\
c[j,S^l_i, S^r_i, T^l_i\setminus \{v\},, T^r_i, r, w]& \mbox{if  $v\in T^l_i$,} \\
c[j,S^l_i, S^r_i, T^l_i, T^r_i\setminus \{v\}, r, w]& \mbox{if $v \in T^r_i$}. \\
  \end{cases}
\end{eqnarray*}
As $v \in X_i$, exactly one of the above cases is applied. 

\medskip

\paragraph*{Introduce edge $(u, v)$ node:} 
Let $i$ be an introduce node of $T$.
Let $S^l_i, S^r_i, T^l_i, T^r_i$ be disjoint subsets of $X_i$ whose union covers $X_i$.
If exactly one of $u$ and $v$ belongs to $S^l_i\cup S^r_i$ (i.e. the other one belongs to $T^l_i\cup T^r_i$),
the edge is included in the cutset.
Suppose otherwise, that is, either $u, v \in S^l_i \cup S^r_i$ or $u, v \in T^l_i \cup T^r_i$. 
If $u$ and $v$ belong to different sets, say $u \in S^l_i$ and $v \in S^r_i$, then
$[S^l_i, S^r_i]$ is not consistent.
Therefore, there is no partial solutions in this case.
To summarize these facts, we have the following:
\begin{eqnarray*}
c[i,S^l_i, S^r_i, T^l_i, T^r_i, r, w]
& = \begin{cases}
	c[j,S^l_i, S^r_i, T^l_i, T^r_i, r-1, w]& \mbox{if $|(S^l_i\cup S^r_i) \cap \{u, v\}| = 1$,} \\
	c[j,S^l_i, S^r_i, T^l_i, T^r_i, r, w] & \mbox{if  $u,v$ are in the same set,}\\
	0 & \mbox{otherwise}.
  \end{cases}
\end{eqnarray*}

\medskip

\paragraph*{Forget $v$ node:}
In a forget node $i$, we just sum up the number of partial solutions:
\begin{eqnarray*} 
c[i,S^l_i, S^r_i, T^l_i, T^r_i, r, w] = &c[j,S^l_i\cup \{v\}, S^r_i, T^l_i, T^r_i, r, w] + c[j,S^l_i, S^r_i\cup \{v\}, T^l_i, T^r_i, r, w]\\
&+c[j,S^l_i, S^r_i, T^l_i\cup \{v\}, T^r_i, r, w] +c[j,S^l_i, S^r_i, T^l_i, T^r_i\cup \{v\}, r, w].
\end{eqnarray*}

\medskip

\paragraph*{Join node:}
Let $i$ be a join node and $j_1$ and $j_2$ its children.
As $X_i = X_{j_1} = X_{j_2}$, it should hold that $S^l_i = S^l_{j_1} = S^l_{j_2}$, $S^r_i = S^r_{j_1} = S^r_{j_2}$, $T^l_i = T^l_{j_1} = T^l_{j_2}$, and $T^r_i = T^r_{j_1} = T^r_{j_2}$.

The size of a partial solution $S_i$ at $i$ is the sum of the size of partial solutions $S_{j_1}$ and $S_{j_2}$ at its children, minus the number of edges from $S_i\cap X_i$ to $X_i\setminus S_i$, since such edges are in both $S_{j_1}$ and $S_{j_2}$. Also, the total weight of $S_i$ is the the sum of the weight of $S_{j_1}$ and $S_{j_2}$ minus the the total weight of $S_i\cap X_i$.

Thus, we have $$c[i,S^l_i, S^r_i, T^l_i, T^r_i, r, w]$$ equals to:
\begin{eqnarray*} 
 \sum_{r_{j_1}+r_{j_2}-\alpha~=r~}   \sum_{w_{j_1}+w_{j_2}-\beta~=~w} c[j_1,S^l_i, S^r_i, T^l_i, T^r_i, r_{j_1}, w_{j_1}] c[{j_2},S^l_i, S^r_i, T^l_i, T^r_i, r_{j_2}, w_{j_2}].
\end{eqnarray*}
where $\alpha$ is the set of edges from $S^l_i\cup S^r_i$ to $T^l_i\cup T^r_i$, and $\beta$ is the sum of the weights of vertices in $S^l_i\cup S^r_i$.

The running time of evaluating the recursive formulas is $O^*(4^{|X_i|})$ for each node $i$.
Therefore, the total running is $O^*(4^{tw})$.
We can also solve  \textsc{Largest Bond} in time $O^*(4^{tw})$ by applying the algorithm for \textsc{Largest $st$-Bond} for all combinations of $s$ and $t$.

\begin{theorem}\label{thm:treewidth_single:Maxmin}
Given a tree decomposition of width $tw$, there is a Monte-Carlo algorithm that solves \textsc{Largest Bond} and \textsc{Largest $st$-Bond} in time $O^*(4^{tw})$. It cannot give false positives and may give false negatives with probability at most 1/2.
\end{theorem}

We can also solve  \textsc{Maximum Connected Cut} and \textsc{Maximum Connected $st$-Cut}. Since it suffices to keep track of consistent cuts of $S$, the running time is $O^*(3^{tw})$.

\begin{theorem}\label{thm:treewidth_single:Connected}
Given a tree decomposition of width $tw$, there is a Monte-Carlo algorithm that solves  \textsc{Maximum Connected Cut} and \textsc{Maximum Connected $st$-Cut}  in time $O^*(3^{tw})$. It cannot give false positives and may give false negatives with probability at most 1/2.
\end{theorem}

By combining Corollary \ref{cor:bounding_tw} and Theorems \ref{thm:treewidth_single:Maxmin} and \ref{thm:treewidth_single:Connected}, we have the following theorem.
\begin{theorem}\label{thm:solution_size:Connected}
There are Monte-Carlo algorithms that solve  \textsc{Largest Bond} and \textsc{Maximum Connected Cut} in time $O^*(16^{k})$ and $O^*(3^{k})$, respectively. It cannot give false positives and may give false negatives with probability at most 1/2.
\end{theorem}

\subsection{Twin-cover}\label{sec:twincover}

Two vertices $u,v$ are called {\em twins} if both $u$ and $v$ have the same closed/open neighbourhood. Moreover, if twins $u,v$ have edge $\{u,v\}$, they are called {\em true twins} and the edge is called a {\em twin edge}. Then a {\em twin-cover} of $G$ is defined as follows.

\begin{definition}[\cite{Ganian2015}]\label{def:twincover}
A set of vertices $X$ is a {\em twin-cover} of $G$ if every edge $\{u,v\}\in E$ satisfies either
\begin{itemize}
\item $u\in X$ or $v\in X$, or
\item $u,v$ are true twins.
\end{itemize}  
The {\em twin-cover number} of $G$, denoted by $tc(G)$, is defined as the size of minimum twin-cover in $G$. 
\end{definition}

An important observation is that the complement of a twin-cover $X$ induces disjoint cliques.
Moreover, for each clique $Z$ of $G[V \setminus X]$, $N(u) \cap X = N(v) \cap X$ for every $u, v \in Z$~\cite{Ganian2015}.

\textsc{Maximum Cut} is FPT when parameterized by  twin-cover number~\cite{Ganian2015}. In this section, we show that \textsc{Maximum Connected Cut} and \textsc{Largest Bond} are also FPT when parameterized by the twin-cover number.

\begin{theorem}\label{thm:twincover:concut}
\textsc{Maximum Connected Cut} can be solved in time $O^*(2^{2^{tc}+tc})$.
\end{theorem}
\begin{proof}
First compute a minimum twin-cover $X$ of $G$ in time $O^*(1.2738^{tc})$~\cite{Ganian2015}.
Now, we have a twin-cover $X$ of size $tc$. Recall that $G[V\setminus X]$ consists of vertex disjoint cliques and for each $u, v \in Z$ in a clique $Z$ of $G[V \setminus X]$, $N(u) \cap X = N(v) \cap X$.

We iterate over all possible subsets $X'$ of $X$ and compute the size of a  maximum cut $[S, V\setminus S]$ of $G$ with $S \cap X = X'$.

If $X' = \emptyset$, exactly one of the cliques of $G[V \setminus X]$ intersects $S$ as $G[S]$ is connected.
Thus, we can compute a maximum cut by finding the best partition for each clique of $G[V \setminus X]$, which can be done in polynomial time.

Suppose otherwise that $X' \neq \emptyset$.
We define a {\em type} of each clique $Z$ of $G[V \setminus X]$.
The type of $Z$, denoted by $T(Z)$, is $N(Z) \cap X$. Note that there are at most $2^{tc}-1$ types of cliques in $G[V \setminus X]$.

For each type of cliques, we guess that $S$ has an intersection with this type of cliques. There are at most $2^{2^tc - 1}$ possible combinations of types of cliques.
Let $\mathcal T$ be the set of types in $G[V \setminus X]$. 
For each guess $\mathcal T' \subseteq \mathcal T$, we try to find a maximum cut $[S, V \setminus S]$ such that $G[S]$ is connected, $S \cap X = X'$, for each $T \in \mathcal T'$, at least one of the cliques of type $T$ has an intersection with $S$, and for each $T \notin \mathcal T'$, every clique of type $T$ has no intersection with $S$.
We can easily check if $G[S]$ will be connected as $S$ contains a vertex of a clique of type $T \in \mathcal T'$.
Consider a clique $Z$ of type $T(Z) = X'' \subseteq X$. 
Since every vertex in $Z$ has the same neighborhood in $X$, we can determine the number of cut edges incident to $Z$ from the cardinality of $S \cap Z$.
More specifically, if $|S \cap Z| = p$, the number of cut edges incident to $Z$ is equal to $p(|Z|-p)+p|X''\cap (X \setminus X')| + (|Z|-p)|X'' \cap X'|$.
Moreover, we can independently maximize the number of cut edges incident to $Z$ for each clique $Z$ of $G[V \setminus X]$.

Overall, for each $X' \subseteq X$ and for each set of types $\mathcal T'$, we can compute a maximum connected cut with respect to $X'$ and $\mathcal T'$ in polynomial time.
Therefore, the total running time is bounded by $O^*(2^{2^tc+tc})$. 
\end{proof}

\begin{theorem}\label{thm:twincover:maxmin}
\textsc{Largest Bond} can be solved in time $O^*(2^{tc}3^{2^{tc}})$.
\end{theorem}
\begin{proof}
We design an $O^*(2^{tc}3^{2^{tc}})$-time algorithm for \textsc{Largest Bond}, where $tc$ is the size of a minimum twin-cover of $G = (V, E)$.
This is quite similar to the one for \textsc{Maximum Connected Cut} developed in this section.
As with an algorithm for \textsc{Maximum Connected Cut}, we first compute a minimum twin-cover $X$ in time $O^*(1.2738^{tc})$~\cite{Ganian2015}.
Then we guess all $2^{tc}$ possible subsets $X'\subseteq X$ and compute the size of maximum cut $(S,V\setminus S)$ of $G$ with $S\cap X=X'$.

If $X'=\emptyset$, exactly one of the cliques of $G[V\setminus X]$ intersects $S$ due to the connectivity of $G[S]$. Thus, we can compute a maximum cut in polynomial time.
Note that $G[V\setminus S]$ is also connected because $X\subseteq  V\setminus S$.
We are also done for the case where $X' = X$ by a symmetric argument.
Thus, in the following, we assume that our guess $X'$ is non-empty and a proper subset of $X$.

For each guess $X' \subseteq X$, we further guess each type of cliques in $G[V \setminus X]$ has an intersection with only $S$, with only $V \setminus S$, or with both $S$ and $V \setminus S$.
For each guess, we can easily check $S$ and $V \setminus S$ will be connected and maximize the size of a cut in polynomial time as in Theorem~\ref{thm:twincover:concut}.
Since there are at most $2^{tc}$ types of cliques in $G[V \setminus X]$, the total running time is $O^*(2^{tc}3^{2^{tc}})$. 
\end{proof}

\section{Infeasibility of polynomial kernels for solution size parameterization}

It is not hard to see that \textsc{Maximum Connected Cut} do not admit a polynomial kernel unless NP $\subseteq$ coNP/poly, since it is trivially or-compositional; at least one of graphs $G_1, G_2, \ldots G_t$ have a connected cut of size at least $k$ if and only if their disjoint union $G_1 \cup G_2 \cup \cdots \cup G_t$ also has a connected cut of size at least $k$.

Regarding to bonds, as seen previously, any bond $\partial(S)$ of a graph $G$ intersects at most one of its block. Thus, an or-composition for {\sc Largest Bond} parameterized by $k$ can be done from the disjoint union of~$\ell$ inputs, by selecting exactly one vertex of each input graph and contracting them into a single vertex.
Now, let $(G_1,k,s_1,t_1),(G_2,k,s_2,t_2),\ldots,(G_\ell,k,s_\ell,t_\ell)$ be $\ell$ instances of {\sc Largest $st$-Bond} parameterized by $k$.
An or-composition for {\sc Largest $st$-Bond} parameterized by $k$ can be done from the disjoint union of $G_1,G_2,\ldots,G_\ell$, by contracting $t_i,s_{i+1}$ into a single vertex, $1\leq i\leq \ell-1$, and setting $s=s_1$ and $t=t_\ell$.

Therefore, the following holds.

\begin{theorem}\label{nokernel}
{\sc Largest Bond}, {\sc Largest $st$-Bond} and {\sc Maximum Connected Cut} do not admit polynomial kernel, unless NP $\subseteq$ coNP/poly.
\end{theorem}

\section{Conclusions}\label{sec:conclusions}

In this work, we present a multivariate analysis on the complexity of computing the largest bond and the maximum connected cut of a graph. Some of our contributions is summarized in Table~\ref{table:contribution}.

Also, we present general reductions that allows us to observe that {\sc Largest Bond} and {\sc Maximum Connected Cut} are \NP-hard for several graph classes for which {\sc Maximum Cut} is \NP-hard. Using this frameworks, we are able to show that {\sc Largest Bond} and {\sc Maximum Connected Cut} on graphs of clique-width $w$ cannot be solved in time $f(w)\times n^{o(w)}$ unless the ETH fails.
Moreover, we show that {\sc Largest Bond} does not admit a constant-factor approximation algorithm, unless $\PP = \NP$, and thus is asymptotically harder to approximate than {\sc Maximum Cut}.

\begin{table}[h]
  \centering
    \caption{The summary of the computational complexity of \textsc{Maximum Cut} and its variants.}
  \begin{tabular}{|c| c c |c c c c| c|} \hline
 & \multicolumn{2}{|c|}{Graph Class} &  \multicolumn{4}{|c|}{Parameter} &poly kernel\cr \hline
 & Split &Planar Bipartite & $cw$ & $tw$ & $tc$ & $k$ & $k$ \cr \hline \hline
 \textsc{Maximum} & NP-c & P & $n^{O(cw)}$ & $2^{tw}$ & $2^{tc}$  & $1.2418^k$ & $O(k)$\cr 
 \textsc{Cut} &   \cite{bodlaender2000}&[trivial] & \cite{fomin2014almost} & \cite{bodlaender2000} &\cite{Ganian2015} &\cite{Raman2007} & \cite{Haglin,mahajan1999parameterizing} \cr \hline
  \textsc{Connected} &  NP-c & NP-c & $n^{O(cw)}$ & $3^{tw}$& $2^{2^{tc}+tc}$ & $2^{O(k)}$ & No\cr 
\textsc{Cut} & [Th.~\ref{thm:split:connected}]   & [Th.~\ref{thm:bipartite:Connected}] & [Th.~\ref{cliquewidth_up_faster}] &[Th.~\ref{thm:treewidth_single:Connected}]  &[Th.~\ref{thm:twincover:concut}] & [Cor.~\ref{thm:singleexpo}]  &  [Th.~\ref{nokernel}]  \cr \hline
     \textsc{Largest} & NP-c & NP-c & $n^{O(cw)}$ & $4^{tw}$& $2^{tc}3^{2^{tc}}$  & $2^{O(k)}$ &No \cr 
\textsc{Bond}   & [Th.~\ref{thm:split:maxmin}] & [Th.~\ref{thm:bipartite:maxmin}] & [Th.~\ref{cliquewidth_up_faster}]& [Th.~\ref{thm:treewidth_single:Maxmin}] &[Th.~\ref{thm:twincover:maxmin}]  & [Cor.~\ref{thm:singleexpo}]  & [Th.~\ref{nokernel}] \cr \hline
  \end{tabular}
  \label{table:contribution}
\end{table}

%
%

\bibliographystyle{spmpsci}      
\bibliography{bond}  

\newpage

\section*{Appendix}
\appendix

\subsection*{Faster algorithms parameterized by clique-width (Theorem \ref{cliquewidth_up_faster})}
In this section, we design faster XP algorithms for both \textsc{Maximum Connected Cut} and \textsc{Largest Bond} when parameterized by clique-width, which run in time $n^{O(w)}$.

Here, we rather use a different graph parameter and its associated decomposition closely related to clique-width.
We believe that this decomposition is more suitable to describe our dynamic programming.

\begin{definition}
    Let $X \subseteq V(G)$. We say that $M \subseteq X$ is a {\em twin-set} of $X$ if for any $v \in V(G) \setminus X$, either $M \subseteq N(v)$ or $M \cap N(v) = \emptyset$ holds. A twin-set $M$ is called a {\em twin-class} of $X$ if it is maximal subject to being a twin-set of $X$. $X$ can be partitioned into twin-classes of $X$.
\end{definition}

\begin{definition}
    Let $w$ be an integer. We say that $X \subseteq V(G)$ is a {\em $w$-module} of $G$ if $X$ can be partitioned into $w$ twin-classes $\{X_1, X_2, \ldots, X_{w}\}$.
    A {\em decomposition tree} of $G$ is a pair of a rooted binary tree $T$ and a bijection $\phi$ from the set of leaves of $T$ to $V(G)$.
    For each node $v$ of $T$, we denote by $L_v$ the set of leaves, each of which is either $v$ or a descendant of $v$.
    The {\em width} of a decomposition tree $(T, \phi)$ of $G$ is the minimum $w$ such that for every node $v$ in $T$, the set $\bigcup_{l \in L_v}\phi(l)$ is a $w_v$-module of $G$ with $w_v \le w$.
    The {\em module-width} of $G$ is the minimum $t$ such that there is a decomposition tree of $G$ of width $w$.
\end{definition}

Rao \cite{Rao2008} proved that clique-width and module-width are linearly related to each other.
Let $cw$ and $mw$ be the clique-width and the module-width of $G$, respectively.
We note that a similar terminology ``modular-width'' has been used in the literature, but module-width used in this paper is different from it.

\begin{theorem}[\cite{Rao2008}]\label{thm:mw-cw}
    For every graph $G$,
    $mw \le cw \le 2mw$.
\end{theorem}
Moreover, given a $w$-expression tree of $G$, we can in time $O(n^2)$ compute a decomposition tree $(T, \phi)$ of $G$ of width at most $w$ and $w_v \le w$ twin-classes of $\bigcup_{l \in L_v}\phi(l)$ for each node $v$ in $T$ \cite{Bui-Xuan2013}.

Fix a decomposition tree $(T, f)$ of $G$ whose width is $w$.
Our dynamic programming algorithm runs over the nodes of the decomposition tree in a bottom-up manner.
For each node $v$ in $T$, we let $\{X^v_1, X^v_2, \ldots, X^v_{w_v}\}$ be the twin-classes of $\bigcup_{l \in L_v}\phi(l)$. From now on, we abuse the notation to denote $\bigcup_{l \in L_v}\phi(l)$ simply by $L_v$.
A tuple of $4w_v$ integers $t= (p_1, \overline{p}_1, p_2, \overline{p}_2, \ldots, p_{w_v}, \overline{p}_{w_v}, c_1, \overline{c}_1, c_2, \overline{c}_2, \ldots, c_{w_v}, \overline{c}_{w_v})$ is {\em valid} for $v$ if it holds that $0 \le p_i, \overline{p}_i \le |X^v_i|$ with $p_i + \overline{p}_i = |X^v_i|$ and $c_i, \overline{c}_i \in \{0, 1\}$ for each $1 \le i \le w_v$.
For a valid tuple $t$ for $v$, we say that a cut $(S, L_v \setminus S)$ of $G[L_v]$ is {\em $t$-legitimate} if for each $1 \le i \le w_v$, it satisfies the following conditions:
\begin{itemize}
    \item $p_i = |S \cap X^v_i|$,
    \item $\overline{p}_i = |(L_v \setminus S) \cap X^v_i|$,
    \item $G[S \cap X^v_i]$ is connected if $c_i = 1$, and
    \item $G[(L_v \setminus S) \cap X^v_i]$ is connected if $\overline{c}_i = 1$.
\end{itemize}
The size of a $t$-legitimate cut is defined accordingly.
In this section, we allow each side of a cut to be empty and the empty graph is considered to be connected.
Our algorithm computes the value $\mc(v, t)$ that is the maximum size of a $t$-legitimate cut for each valid tuple $t$ and each node $v$ in the decomposition tree. 

\medskip
\noindent \textbf{Leaves (Base step):} 
For each valid tuple $t$ for a leaf $v$, $\mc(v, t) = 0$. 
Note that there is only one twin-class $X^v_1 = \{v\}$ for $v$ in this case. 

\medskip
\noindent \textbf{Internal nodes (Induction step):} Let $v$ be an internal node of $T$ and let $a$ and $b$ be the children of $v$ in $T$.
Consider twin-classes $\CX^v = \{X^v_1, X^v_2, \ldots, X^v_{w_v}\}$, $\CX^a = \{X^a_1, X^a_2, \ldots, X^a_{w_a}\}$, and $\CX^b = \{X^b_1, X^b_2, \ldots, X^b_{w_b}\}$ of $L_v$, $L_a$, and $L_b$, respectively. Note that $\CX^a \cup \CX^b$ is a partition of $L_v$.

\begin{observation}\label{obs:coarse}
    $\CX^v$ is a partition of $L_v$ coarser than $\CX^a \cup \CX^b$.
\end{observation}

To see this, consider an arbitrary twin-class $X^a_i$ of $L_a$. By the definition of twin-sets, for every $z \in V(G) \setminus L_a$, either $X^a_i \subseteq N(z)$ or $X^a_i \cap N(z) = \emptyset$ holds. Since $V(G) \setminus L_v \subseteq V(G) \setminus L_a$, $X^a_i$ is also a twin-set of $L_v$, which implies $X^a_i$ is included in some twin-class $X^v_j$ of $L_v$. This argument indeed holds for twin-classes of $L_b$. Therefore, we have the above observation.

The intuition of our recurrence is as follows.
By Observation~\ref{obs:coarse}, every twin-class of $L_v$ can be obtained by merging some twin-classes of $L_a$ and of $L_b$.
This means that every $t_v$-legitimate cut of $G[L_v]$ for a valid tuple $t_v$ for $v$ can be obtained from some $t_a$-legitimate cut and $t_b$-legitimate cut for valid tuples for $a$ and $b$, respectively.
Moreover, for every pair of twin-classes $X^a_i$ of $L_a$ and $X^b_j$ of $L_b$, either there are no edges between them or every vertex in $X^a_i$ is adjacent to every vertex in $X^b_j$ as $X^a_i$ is a twin-set of $L_v$.
Therefore, the number of edges in the cutset of a cut $(S, L_v \setminus S)$ between $X^a_i$ and $X^b_j$ depends only on the cardinality of $X^a_i \cap S$ and $X^b_j \cap S$ rather than actual cuts $(S \cap X^a_i, (L_a \setminus S) \cap X^a_i)$ and $(S \cap X^b_i, (L_b \setminus S) \cap X^b_i)$.

Now, we formally describe this idea. Let $X^v$ be a twin-class of $L_v$.
We denote by $I_a(X^v)$ (resp. $I_b(X^v)$) the set of indices $i$ such that $X^a_i$ (resp. $X^b_i$) is included in $X^v$ and by $\CX^a(X^v)$ (resp. $\CX^b(X^v)$) the set $\{X^a_i : i \in I_a(X^v)\}$ (resp. $\{X^b_i : i \in I_b(X^v)\}$).
For $X^a \in \CX^a(X^v)$ and $X^b \in \CX^a(X^v)$, we say that $X^a$ is adjacent to $X^b$ if every vertex in $X^a$ is adjacent to every vertex in $X^b$ and otherwise $X^a$ is not adjacent to $X^b$. 
This adjacency relation naturally defines a bipartite graph whose vertex set is $\CX^a(X^v) \cup \CX^b(X^v)$.
We say that a subset of twin-classes of $\CX^a(X^v) \cup \CX^b(X^v)$ is {\em non-trivially connected} if it induces a connected bipartite graph with at least twin-classes.
Let $S \subseteq X^v$.
To make $G[S]$ (and $G[X^v \setminus S]$) connected, the following observation is useful.
\begin{observation}\label{obs:connect}
    Suppose $S \subseteq X^v$ has a non-empty intersection with at least two twin-classes of $\CX^a(X^v) \cup \CX^b(X^v)$.
    Then, $G[S]$ is connected if and only if the twin-classes having a non-empty intersection with $S$ are non-trivially connected.
\end{observation}
This observation immediately follows from the fact that every vertex in a twin-class is adjacent to every vertex in an adjacent twin-class and is not adjacent to every vertex in a non-adjacent twin-class.

Let $t_v = (p^v_1, \overline{p}^v_1, \ldots, p^v_{w_v}, \overline{p}^v_{w_v}, c^v_1, \overline{c}^v_2, \ldots, c^v_{w_v}, \overline{c}^v_{w_v})$ be a valid tuple for $v$.
For notational convenience, we use ${\bf p}^v$ to denote $(p^v_1, \overline{p}^v_1, \ldots, p^v_{w_v}, \overline{p}^v_{w_v})$ and ${\bf c}^v$ to denote $(c^v_1, \overline{c}^v_2, \ldots, c^v_{w_v}, \overline{c}^v_{w_v})$ for each node $v$ in $T$.
For valid tuples $t_a = ({\bf p}^a, {\bf c}^a)$ for $a$ and $t_b = ({\bf p}^b, {\bf c}^b)$ for $b$,
we say that {\em $t_v$ is consistent with the pair $(t_a, t_b)$} if for each $1\le i \le w_v$, 
\begin{itemize}
    \item[C1] $p^v_i = \sum_{j \in I_a(X^v_i)} p^a_j + \sum_{j \in I_b(X^v_i)} p^b_j$;
    \item[C2] $\overline{p}^v_i = \sum_{j \in I_a(X^v_i)} \overline{p}^a_j + \sum_{j \in I_b(X^v_i)} \overline{p}^b_j$;
    \item[C3] if $c^v_i = 1$, either (1) $\{X^a_j : j \in I_a(X^v), p^a_j > 0\} \cup \{X^b_j : j \in I_b(X^v), p^b_j > 0\}$ is non-trivially connected or (2) exactly one of $\{p^s_j : s \in \{a, b\}, 1 \le j \le w_s\}$ is positive, say $p^s_j$, and $c^s_j = 1$;
    \item[C4] if $\overline{c}^v_i = 1$, either (1) $\{X^a_j : j \in I_a(X^v), \overline{p}^a_j > 0\} \cup \{X^b_j : j \in I_b(X^v), \overline{p}^b_j > 0\}$ is non-trivially connected or (2) exactly one of $\{\overline{p}^s_j : s \in \{a, b\}, 1 \le j \le w_s\}$ is positive, say $\overline{p}^s_j$, and $\overline{c}^s_j = 1$.
\end{itemize}

\begin{lemma}\label{lem:cw:rec}
    \[
        \mc(v, t_v) = \max_{t_a, t_b} \left(\mc(a, t_a) + \mc(b, t_b) + \sum_{\substack{X^a_i \in \CX^a, X^b_j \in \CX^b\\
        X^a_i, X^b_j: \text{adjacent}}} (p^a_i\overline{p}^b_j + p^b_j\overline{p}^a_i)  \right),
    \]
    where the maximum is taken over all consistent pairs $(t_a, t_b)$.
\end{lemma}

\begin{proof}
We first show that the left-hand side is at most the right-hand side.
Suppose $(S, L_v \setminus S)$ be a $t_v$-legitimate cut of $G[L_v]$ whose size is equal to $\mc(v, t_v)$.
Let $S_a = S \cap L_a$ and $S_b = S \cap L_b$.
We claim that $(S_a, L_a \setminus S_a)$ is a $t_a$-legitimate cut of $G[L_a]$ for some valid tuple $t_a$ for $a$. 
This is obvious since we set $p^a_i = |S_a \cap X^a_i|$, $\overline{p}^a_i = |(L_a \setminus S_a) \cap X^a_i|$, $c^a_i = 1$ if $G[S_a \cap X^a_i]$ is connected, and $c^a_i = 1$ if $G[(L_a \setminus S_a) \cap X^a_i]$ is connected, which yields a valid tuple $t_a$ for $a$.
We also conclude that  $(S_b, L_b \setminus S_b)$ is a $t_b$-legitimate cut of $G[L_b]$ for some valid tuple $t_b$ for $b$.
Moreover, the number of cut edges between twin-class $X^a_i$ of $L_a$ and twin-class $X^b_j$ of $L_b$ is $|S_a \cap X^a_i|\cdot |(L_b\setminus S_b) \cap X^b_j| + |S_b \cap X^b_j|\cdot|(L_b \setminus S_a)\cap X^a_i| = p^a_i\overline{p}^b_j + p^b_j\overline{p}^a_i$ if $X^a_i$ and $X^b_j$ is adjacent, zero otherwise.
Therefore, the left-hand side is at most the right-hand side.

To show the converse direction, suppose $(S_a, L_a \setminus S_a)$ is a $t_a$-legitimate cut of $G[L_a]$ and $(S_b, L_b \setminus S_b)$ is a $t_b$-legitimate cut of $G[L_b]$, where $t_v$ is consistent with $(t_a, t_b)$ and the sizes of the cuts are $\mc(a, t_a)$ and $\mc(b, t_b)$, respectively.
We claim that $(S_a \cup S_b, L_v \setminus (S_a \cup S_b))$ is a $t_v$-legitimate cut of $G[L_v]$.
Since $t_v$ is consistent with $(t_a, t_b)$, for each $1 \le i \le w_v$, we have $p^v_i = \sum_{j \in I_a(X^v_i)} p^a_j + \sum_{j \in I_b(X^v_i)} p^b_j = \sum_{1 \le j \le w_a}|S_a \cap X^i_v| + \sum_{1 \le j \le w_b}|S_b \cap X^i_v| = |(S_a \cup S_b) \cap X^i_v|$. Symmetrically, we have $\overline{p}^i = |(L_v \setminus (S_a \cup S_b)) \cap X^v_i|$.
If $c^v_i = 1$, by condition C3 of the consistency, either (1) $\{X^a_j : j \in I_a(X^v), p^a_j > 0\} \cup \{X^b_j : j \in I_b(X^v), p^b_j > 0\}$ is non-trivially connected or (2) exactly one of $\{p^s_j : s \in \{a, b\}, 1 \le j \le w_s\}$ is positive, say $p^s_j$, and $c^s_j = 1$. 
If (1) holds, by Observation~\ref{obs:connect}, $G[(S_a \cap S_b) \cap X^i_v]$ is connected. Otherwise, as $c^s_j = 1$, $G[S_s \cap X^i_v] = G[(S_a \cup S_b) \cap X^v_i]$ is also connected.
By a symmetric argument, we conclude that $G[(L_v \setminus (S_a \cup S_b)) \cap X^i_v]$ is connected if $\overline{c}^v_i = 1$. Therefore the cut $(S_a \cup S_b, L_v \setminus (S_a \cup S_b))$ is $t_v$-legitimate.
Since the cut edges between two twin-classes of $L_a$ is counted by $\mc(a, t_a)$ and those between two twin-classes of $L_v$ is counted by $\mc(b, t_b)$. 
Similar to the forward direction, the number of cut edges between a twin-class of $L_a$ and a twin-class of $L_b$ can be counted by the third term in the right-hand side of the equality.
Hence, the left-hand side is at least the right-hand side. 
\end{proof}


\noindent{\em Proof of Theorem \ref{cliquewidth_up_faster}.}
    From a $w$-expression tree of $G$, we can obtain a decomposition tree $(T, \phi)$ of width at most $w$ in $O(n^2)$ time using Rao's algorithm \cite{Rao2008}.
    Based on this decomposition, we evaluate the recurrence in Lemma~\ref{lem:cw:rec} in a bottom-up manner. The number of valid tuples for each node of $T$ is at most $4^wn^w$. For each internal node $v$ and for each valid tuple $t_v$ for $v$, we can compute $\mc(v, t_v)$ in $(4^wn^w)^2n^{O(1)}$ time. Overall, the running time of our algorithm is $n^{O(w)}$.
    Let $r$ be the root of $T$.
    For \textsc{Maximum Connected Cut}, by the definition of legitimate cuts, we should take the maximum value among $\mc(r, (i, n-i, 1, j))$ for $1 \le i < n$ and $j \in \{0, 1\}$. Note that as $L_v$ has only one twin-class, the length of valid tuples is exactly four. For \textsc{Largest Bond}, we should take the maximum value among $\mc(r, (i, n-i, 1, 1))$ for $1 \le i < n$. 


\end{document}